\documentclass[11pt,reqno]{amsart}
\allowdisplaybreaks[4]
\usepackage{graphicx}
\usepackage{subfigure}
\usepackage{epsfig}
\usepackage{amssymb}
\usepackage{mathtools}
\usepackage{amsmath,xypic}
\usepackage{cite,color}
\usepackage{cite}
\usepackage{amsmath}
\usepackage{tikz}
\usepackage{epic,eepic,amsmath,amsthm,amssymb}
\usepackage{epsfig,cite,indentfirst,oldgerm}
\usepackage{multicol,boxedminipage}
\usepackage{epic,eepic,amsmath,amsthm,amssymb}
\usepackage{epsfig,cite,indentfirst,oldgerm}
\usepackage{multicol,boxedminipage}
\usepackage{array,caption}
\usepackage{enumerate}
\usepackage{booktabs}
\numberwithin{equation}{section}

\newtheorem{ca}{Figure}
\newtheorem{theorem}{Theorem}[section]
\newtheorem{definition}[theorem]{Definition}
\newtheorem{proposition}[theorem]{Proposition}

\newtheorem{lemma}[theorem]{Lemma}
\newtheorem{remark}[theorem]{Remark}

\newcommand{\be}{\begin{equation}}
	\newcommand{\ee}{\end{equation}}
\newcommand{\bea}{\begin{eqnarray}}
	\newcommand{\eea}{\end{eqnarray}}
\newcommand{\ba}{\begin{array}}
	\newcommand{\ea}{\end{array}}
\newcommand{\bean}{\begin{eqnarray*}}
	\newcommand{\eean}{\end{eqnarray*}}

\def\lprod{\mathop{\prod{\mkern-29.5mu}{\mathbf\longleftarrow}}}
\def\rprod{\mathop{\prod{\mkern-28.0mu}{\mathbf\longrightarrow}}}
\begin{document}
	\title{Generalized $i$-boson model and boxed BUC plane partitions}
	\author{Shengyu Zhang, Denghui Li and  Zhaowen Yan$^*$}
	
	\dedicatory {School of Mathematical Sciences, Inner Mongolia University, \\ Inner Mongolia Key Laboratory of Mathematical Modeling and Scientific Computing,
		\\ Hohhot, Inner Mongolia 010021,  P. R. China }
	
	\thanks{*Corresponding author. Email: yanzw@imu.edu.cn (Z.W. Yan)}
	
	\begin{abstract}
This paper is devoted to investigating the relation between the   generalized $i$-boson model and  boxed BUC plane partitions. The representation of  the generalized $i$-boson algebra and the actions of the monodromy matrix operators on  basis vectors  have  been studied. We also consider the actions of neutral fermion vertex operators on state vectors  in terms of the neutral fermionic Fock space. With the help of the scalar product of the   generalized $i$-boson model, the generating function for  boxed BUC plane partitions is derived which can  be represented as the products of Schur $Q$-functions. Moreover, the generating function for BUC plane partitions with the double scaling limit is presented. \\
\textbf{Keywords}: Quantum Inverse Scattering Method; quantum integrable models;   generalized $i$-boson  model;  generating functions of plane partitions\\
\textbf{ Mathematics Subject Classifications (2000)}: 17B80, 35Q55, 37K10
\end{abstract}
\maketitle
\tableofcontents
\section{Introduction}
Great interest has been shown in quantum integrable models in mathematical physics, motivated by close relations with quantum groups\cite{quantum-groups}, conformal field theory\cite{conformal-field},   Yang-Baxter equation\cite{Yang-Baxter-1,Yang-Baxter-2}, especially in symmetric functions\cite{Mac} and classical integrable models\cite{BKP-1,BKP-2,V.G.,Jimbo1981}. The Quantum Inverse Scattering Method (QISM)\cite{QISM-1,QISM-2}  provide a fundamental tool for solving quantum integrable models. The $i$-boson model solved by the QISM is the special case of the $q$-boson model\cite{q-boson} in the limit $q \to i=\sqrt{-1}$.
Wheeler\cite{boshilunwen} found that Bethe eigenvectors of the $i$-boson model are  $\tau$ functions of the BKP hierarchy. Jing et al.\cite{i-boson} explored the actions of  operators formed by neutral fermions on state vectors and the correlation function of the $i$-boson model. The phase model as the limit of $q \to \infty$  in the $q$-boson model was discussed\cite{q-boson,B97}. According to the QISM, Tsilevich\cite{q-boson-hall-little} realized the phase model and the $q$-boson model in terms of the algebra of symmetric functions.  It is showed that the correlation function for the phase model has intimate connection with the  solution of the Toda hierarchy\cite{Q-boson}. The five-vertex model  and the four-vertex model  are special cases of the six-vertex model\cite{six}, in which  two vertices and one of vertices are frozen out, respectively. The relations between the limiting  forms of the scalar product for the five-vertex model and complete symmetric polynomials have been well discussed\cite{five}.

Plane partitions\cite{MacMahon}  have wide application in combinatorics\cite{C1,C2}, statistical mechanics\cite{Statistical-1,Statistical-2}, integrable systems\cite{ferm} and quantum field theory\cite{bosonic}. Based on the correlation functions of Schur process and shifted Schur process, the generating functions for  plane partitions and strict plane partitions  have been  studied\cite{Schur-process,shifted-Schur-process}. Meanwhile, Foda et al.\cite{Foda06,Foda09} presented the generating functions for  plane partitions and strict plane partitions based upon the fermionic description of the KP  and BKP hierarchies.  Vuleti\'{c}\cite{2-parameter} developed $1$-parameter and $2$-parameter extensions of the generating function for plane partitions related to Macdonald's symmetric functions. By means of charged $t$-fermions and  the deformed boson-fermion correspondence\cite{Jing-B,Jing-F}, $1$-parameter extension of the generating function for plane partitions has been derived\cite{1-parameter}. The generating functions for Unicersal Character (UC) and  Unicersal Character of B type (BUC) plane partitions were established by  the fermion calculus approach\cite{UC-BUC}.
	It is worth noting that the connection between the quantum integrable models solved by the QISM and boxed plane partitions  has been explored.
	With the help of the scalar product of the phase model,
	the generating function of boxed plane partitions and boxed skew plane partitions were analyzed\cite{Boxed-q-boson-model,Boxed-skew-phase-model}.
	Bogoliubov \cite{Four-vertex-Model} developed the generating  function for strict boxed plane partitions for the four-vertex model.
Factorization formulas for  partition functions of the rational six vertex model  under trapezoid boundary were introduced\cite{Six-vertex-Model}.
	Recently, the generating function of boxed UC plane partitions has been established\cite{t-g-phase-model-boxed-UC}. The main objective of this paper is to discuss the  relation  between the   generalized $i$-boson model and boxed BUC plane partitions.

	The paper is organized as follows.
	Section $2$ is devoted to the presentation of some fundamental facts about
	plane partitions and neutral fermions.
	The representation of   generalized $i$-boson algebras and the   generalized $i$-boson model  have been constructed in Section $3$.
	In Section $4$, we  define  maps from the  space of quantum
	state vectors  to the neutral fermionic Fock space and study the actions of the monodromy matrix operators on  basis vectors.
	In Section $5$,  the generating function for  boxed BUC plane partitions  be expressed  as the products of Schur $Q$-functions by means of the scalar product of the   generalized $i$-boson model.
	The core of Section $6$  is the investigation of the double scaling limit of the scalar product. The last Section is conclusions and discussions.

	\section{Preliminaries}
	In this section, we mainly retrospect some basic facts about plane partitions, neutral fermions and  state vectors corresponding to  strict $2$-partitions in $\tilde{\mathcal{F}}$ and $\tilde{\mathcal{F}}^*$\cite{Mac,BUC,boshilunwen,UC-BUC}.
	\subsection{Plane partitions}
	
	A (strict) partition  $\mu=(\mu_{1},\mu_{2},\ldots)$ is (strict) decreasing non-negative integers sequence with the weight $|\mu|=\sum\limits_{i\geq1}\mu_i$.
	The partition can be rewritten as
	\begin{align}
		\mu=(1^{m_{1}}2^{m_{2}}...),
	\end{align}
	where $m_{i}=m_{i}(\mu)$ represents the number of  times that the element $i$  appears  in  the partition $\mu$, and $m_{i}$ equals either $0$ or $1$ for strict partitions.

	A pair of partitions $\mu$ and $\nu=(\nu_{1},\nu_{2},\ldots)$ satisfy
	\begin{eqnarray}\
		\mu_{1}\geq\nu_{1}\geq\mu_{2}\geq\nu_{2}\geq\ldots,
	\end{eqnarray}
	we call that $\nu$ interlaces $\mu$ and denote it as $\mu\succ\nu$.
	Meanwhile, for (strict) $2$-partition $(\tilde{\chi}_1)= (\tilde{\mu}^1,\tilde{\nu}^1 )$ and $(\tilde{\chi}_2)= (\tilde{\mu}^2,\tilde{\nu}^2 )$, there is an interlacing relation
	\begin{align}
		(\tilde{\chi}_1)\succ(\tilde{\chi}_2)  \quad\Longleftrightarrow\quad    \tilde{\mu}^1\succ\tilde{\mu}^2 \quad \mathrm{and} \quad \tilde{\nu}^1 \succ \tilde{\nu}^2.
	\end{align}
	
	Plane partition, also known as a $3$-dimensional $(3D)$ Young diagram, involves adding non-negative integers $\pi_{ij}$ to $i$-th row and $j$-th column of the Young diagram, which satisfies
	\begin{align}
		\pi_{ij}\geq\pi_{(i+1)j}, \quad\pi_{ij}\geq\pi_{i(j+1)}, \quad \lim_{i\to\infty}\pi_{ij}=\lim_{j\to\infty}\pi_{ij}=0,\quad \text{for}~i,j\geq1.
	\end{align}
	Slicing diagonally yields a series of interlacing partitions
	\begin{eqnarray}\label{interlacing-partitions}
		 \emptyset=\pi_{-M}\prec\cdots\prec\pi_{-2}\prec\pi_{-1}\prec\pi_{0}\succ\pi_{1}\succ\pi_{2}\succ\cdots\succ\pi_{N}=\emptyset,
	\end{eqnarray}
	where the plane partition $\pi$
	is denoted as  $\pi=(\ldots,\pi_{-1},\pi_0,\pi_1,\ldots)$ with the weight  $|\pi|=\sum\limits_{i=-M}^{N}|\pi_i|$. The partition $\pi_{i}$ is represented as
	\begin{eqnarray}
		 \pi_i=\begin{cases}(\pi_{1(i+1)},\pi_{2(i+2)},\pi_{3(i+3)},\ldots)&\quad\text{for}~i\geq0,\\(\pi_{(-i+1)1},\pi_{(-i+2)2},\pi_{(-i+3)3},\ldots)&\quad\text{for}~i\leq-1.\end{cases}
	\end{eqnarray}
	If $\pi_{i}$ is a strict partition, then $\pi$ is called a strict plane partition.
	\begin{center}
		\setlength{\unitlength}{0.0012cm}
		\renewcommand{\dashlinestretch}{30}
		\hspace{-4.5cm}\begin{picture}(4800, 3600)(-3000, 0)
			\thicklines			
			\path(0000,0600)(0600,0600)
			\path(0000,1200)(1800,1200)
			\path(0000,1800)(2400,1800)
			\path(0000,2400)(3000,2400)
			\path(0000,3000)(0000,0600)
			\path(0000,3000)(3000,3000)
			\path(0600,3000)(0600,0600)
			\path(1200,3000)(1200,1200)
			\path(1800,3000)(1800,1200)			
			\path(2400,3000)(2400,1800)
			\path(3000,3000)(3000,2400)
			\put(0300,0900){1}
			\put(0300,1500){3}
			\put(0300,2100){4}
			\put(0300,2700){5}
			\put(0900,1500){1}
			\put(0900,2100){2}
			\put(0900,2700){4}
			\put(1500,1500){1}
			\put(1500,2100){2}
			\put(1500,2700){3}
			\put(2100,2100){1}
			\put(2100,2700){2}
			\put(2700,2700){1}
			\path(0000,1200)(0900,0300)
			\path(0000,1800)(0900,0900)
			\path(0000,2400)(1500,0900)
			\path(0000,3000)(2100,0900)
			\path(0600,3000)(2100,1500)
			\path(1200,3000)(2700,1500)
			\path(1800,3000)(2700,2100)
			\path(2400,3000)(3300,2100)
			\put(0900,0600){$\pi_{-2}$}
			\put(0900,0000){$\pi_{-3}$}
			\put(1500,0600){$\pi_{-1}$}
			\put(2100,1200){$\pi_{ 1}$}
			\put(2100,0600){$\pi_{ 0}$}
			\put(2700,1800){$\pi_{ 3}$}
			\put(2700,1200){$\pi_{ 2}$}
			\put(3300,1800){$\pi_{ 4}$}
\end{picture}
\begin{ca}\label{strict-2D}
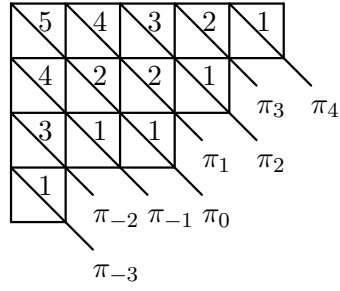
A $2$-dimensional view of the strict plane partition.  The sequence of values covered in the slice is the corresponding partition. In particular, $\pi_{0}=(5,2,1)$ and $|\pi|=30$.
\end{ca}
\end{center}
	
We refer to the (strict) partition $\mu$ confined within rectangle $[N,M]$  with $N$ rows and $M$ columns  as (strict) boxed partitions and record it as $\mu \in [N,M]$. Set $(\tilde{\chi})= (\tilde{\mu},\tilde{\nu} )$, the (strict) $2$-partition $\tilde{\chi}$ is called (strict) boxed $2$-partition if and only if $\tilde{\mu}$ and $\tilde{\nu}$ are (strict) boxed partitions.
The symbol $\pi  \in [N,L,M]$ expresses that the (strict)  plane partition $\pi$ is limited to an $N\times L\times M$ box, which is called as a (strict) boxed plane partition and $N$, $L$ and $M$ represent row, column and height respectively.

The connected regions in plane partition $\tilde{\pi}$ form paths, and the total number of paths is denoted by $2^{p(\tilde{\pi})}$.
	\begin{center}
		\begin{minipage}{3in}	
			\setlength{\unitlength}{0.0012cm}
			\renewcommand{\dashlinestretch}{30}
			\hspace{-2cm}\begin{picture}(4500, 3200)(-3000, 500)
				\thicklines			
				\path(0000,0600)(0600,0600)
				\path(0000,1200)(0600,1200)
				\path(1800,1200)(0600,1200)
				\path(0000,1800)(0600,1800)
				\path(0600,1800)(1800,1800)
				\path(1800,1800)(2400,1800)						
				\path(0000,2400)(3000,2400)
				\path(0000,3000)(0000,0600)
				\path(0000,3000)(3000,3000)
				\path(0600,3000)(0600,0600)
				\path(1200,3000)(1200,2400)
				\path(1800,3000)(1800,1200)			
				\path(2400,3000)(2400,1800)
				\path(3000,3000)(3000,2400)
				\put(0230,0800){1}
				\put(0230,1400){3}
				\put(0230,2000){4}
				\put(0230,2600){5}
				\put(1100,1400){1}
				\put(1100,2000){2}
				\put(0830,2600){4}
				\put(1430,2600){3}
				\put(2030,2100){1}
				\put(2030,2600){2}
				\put(2630,2600){1}
			\end{picture}
			\begin{ca}\label{strict-2D-1}
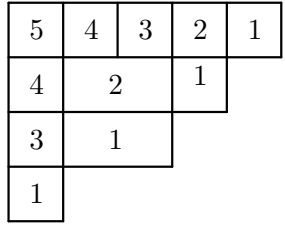
				All  paths in Figure {\bf \ref{strict-2D}}, $p(\tilde{\pi})=11$.
			\end{ca}
		\end{minipage}
	\end{center}

	\begin{lemma}\label{strict-pi}
		Let the number of non-zero elements in $\tilde{\mu}$ be denoted as $l(\tilde{\mu})$, and the number of elements in $\tilde{\mu}$ but not in $\tilde{\nu}$ be represented as $\#(\tilde{\mu}|\tilde{\nu})$.
		The strict plane partition $\tilde{\pi}=(\ldots,\tilde{\pi}_{-1},\tilde{\pi}_0,\tilde{\pi}_1,\ldots)$	 satisfies
		\begin{align}\label{path}
			2^{-l(\tilde{\pi}_0)} \prod_{i=1}^{N} 2^{\#(\tilde{\pi}_{-i+1} | \tilde{\pi}_{-i})} \prod_{j=1}^{N} 2^{\#(\tilde{\pi}_{j-1} | \tilde{\pi}_j)} = 2^{p(\tilde{\pi})}.
		\end{align}
		
	\end{lemma}

	The BUC plane partition $\tilde{\Pi}=(\ldots,\tilde{\chi}_{-1},\tilde{\chi}_0,\tilde{\chi}_1,\ldots)$ is given by
	\begin{eqnarray}\label{BUC-interlacing-partitions} \emptyset=\tilde{\chi}_{-M}\prec\cdots\prec\tilde{\chi}_{-2}\prec\tilde{\chi}_{-1}\prec\tilde{\chi}_{0}\succ\tilde{\chi}_{1}\succ\tilde{\chi}_{2}\succ\cdots\succ\tilde{\chi}_{N}=\emptyset.
	\end{eqnarray}
	The generating function for BUC plane partitions is  written as
	\begin{align}\label{UC-generating-function}
		\sum_{\begin{smallmatrix}\pi^{1}\text{ and }\pi^{2}\text{ are}\\\text{plane partitions}\end{smallmatrix}}
		p^{|\pi^{1}|}q^{|\pi^{2}|}
		=\prod_{n=1}^\infty\left(\frac1{1-p^n}\right)^n\prod_{m=1}^\infty\left(\frac1{1-q^m}\right)^m.
	\end{align}

	\subsection{Neutral fermions and state vectors}
	Consider neutral fermions sets $\{\phi_{n}\}_{n\in\mathbf{Z}}$ and $\{\bar{\phi}_{m}\}_{n\in\mathbf{Z}}$ generated by the relations
	\begin{align}
		\phi_{n}=\psi_{n}+(-1)^{n}\psi_{-n}^{*},\quad
		\bar{\phi}_{m}=\gamma_{m}+(-1)^{m}\gamma_{-m}^{*},
	\end{align}
	where charged fermions $\psi_{n},\psi_{n}^{*},\gamma_{m},\gamma_{m}^{*}$ satisfy
	\begin{align}
		&[\psi_m,\psi_n]_{+} = [\psi_m^*,\psi_n^*]_{+} = 0,[\psi_m,\psi_n^*]_{+} = \delta_{m+n,0},\\
		&[\gamma_m,\gamma_n]_{+} = [\gamma_m^*,\gamma_n^*]_{+} = 0,[\gamma_m,\gamma_n^*]_{+} = \delta_{m+n,0},\\
		&[\psi_{m},\gamma_{n}]=[\psi_{m},\gamma_{n}^{*}]=[\psi_{m}^{*},\gamma_{n}]=[\psi_{m}^{*},\gamma_{n}^{*}]=0,\\
		&\psi_n^2=\psi_n^{*2}=\gamma_m^2=\gamma_m^{*2}=0,
	\end{align}
and
	\begin{align}
		[A, B] = AB - BA, \quad [A, B]_+ = AB + BA, \quad 	 \delta_{m+n,0}=\begin{cases}1,&m+n=0,\\0,&\mathrm{otherwise}.\end{cases}
	\end{align}
	The algebra $\widetilde{\mathcal{A}}$ over
	$\mathbb{C}$ is generated by $\{\phi_{n}\}_{n\in\mathbf{Z}}$ and $\{\bar{\phi}_{m}\}_{m\in\mathbf{Z}}$ satisfying
	\begin{align}\label{commutation relations 1.2}
		 [\phi_{m},\phi_{n}]_{+}=[\bar{\phi}_{m},\bar{\phi}_{n}]_{+}=2(-1)^{m}\delta_{m+n,0},\quad[\phi_{m},\bar{\phi}_{n}]=0,\quad \phi_{m}^{2}=\bar{\phi}_{m}^{2}=\delta_{m,0}.
	\end{align}
	
	Introduce the neutral fermionic Fock space $\mathcal{\widetilde{F}}$ and the dual Fock space $\mathcal{\widetilde{F}}^{*}$,
	\begin{align}
		\mathcal{\widetilde{F}}&\stackrel{\mathrm{def}}{=}\widetilde{\mathcal{A}}\cdot| 0\rangle=\left\{a| 0\rangle|a\in\widetilde{\mathcal{A}}\right\},\notag\\
		\mathcal{\widetilde{F}}^{*}&\stackrel{\mathrm{def}}{=}\langle0|\cdot\widetilde{\mathcal{A}}=\{\langle0|a\mid a\in\widetilde{\mathcal{A}}\},
	\end{align}
	where the vacuum state $|0\rangle$ and the dual vacuum state $\langle0|$ are  denoted as
	\begin{align}
		\phi_n|0\rangle&=\bar{\phi}_n|0\rangle=0\quad\text{for } n<0,\\
		\langle0|\phi_n&=\langle0|\bar{\phi}_n=0\quad\text{for } n>0.
	\end{align}

	Define the  bilinear pairing  $\mathcal{\widetilde{F}}^*\times\mathcal{\widetilde{F}}\longmapsto  \mathbb{C}$ as
	\begin{align}
		(\langle0|a,b|0\rangle)  &\longmapsto  \langle0|a\cdot b|0\rangle=\langle ab\rangle,
	\end{align}
	with $\langle0|0\rangle=1$.
	The operators $\lambda_{m}^{'}$ and $\bar{\lambda}_{m}^{'}$ ($m\in\mathbf{N}_{\mathrm{odd}}$) are expressed as
	\begin{align}
		 \lambda_{m}^{'}=\frac{1}{4}\sum_{j\in\mathbf{Z}}(-1)^{j}\phi_{j}\phi_{-j-m},\quad\bar{\lambda}_{m}^{'}=\frac{1}{4}\sum_{j\in\mathbf{Z}}(-1)^{1}\bar{\phi}_{j}\bar{\phi}_{-j-m}.
	\end{align}
	Then
	\begin{align}\label{commutation relations 2.2}
		[\lambda_{m}^{'},\phi_{n}]=\phi_{n-m},\quad[\bar{\lambda}_{m}^{'},\bar{\phi}_{n}]=\bar{\phi}_{n-m},\quad
		[\lambda_{m}^{'},\lambda^{'}_{n}]=[\bar{\lambda}_{m}^{'},\bar{\lambda}^{'}_{n}]=\frac{m}{2}\delta_{m+n,0},
	\end{align}
	and
	\begin{align}
		\lambda^{'}_m|0\rangle=\bar{\lambda}^{'}_m|0\rangle=0\quad\text{if } m>0.
	\end{align}
	
	For convenience, set
	\begin{align}
		\phi_{n}^{*}&=(-1)^{n}\phi_{-n}=\psi_{n}^{*}+(-1)^{n}\psi_{-n},\\
		\bar{\phi}_{m}^{*}&=(-1)^{m}\bar{\phi}_{-m}
		=\gamma_{m}^{*}+(-1)^{m}\gamma_{-m}.
	\end{align}
	The  state vectors in $\tilde{\mathcal{F}}$ and $\tilde{\mathcal{F}}^*$ identified by $\tilde{\mu}^{j}=(\tilde{\mu}_{1}^{j},\tilde{\mu}_{2}^{j},\cdots,\tilde{\mu}_{2r_{j}}^{j})$ are defined as
	\begin{align}
		\langle\tilde{\mu}^{2},\tilde{\mu}^{1}| &= \langle 0|
		\bar{\phi}_{\tilde{\mu}^{2}_{2r_{2}}}^* \dots \bar{\phi}_{\tilde{\mu}^{2}_1}^*
		\phi_{\tilde{\mu}^{1}_{2r_{1}}}^* \dots \phi_{\tilde{\mu}^{1}_1}^*, \label{o-left-state}\\
		|\tilde{\mu}^{1},\tilde{\mu}^{2}\rangle &=
		\phi_{\tilde{\mu}^{1}_1} \dots  \phi_{\tilde{\mu}^{1}_{2r_{1}}}
		\bar{\phi}_{\tilde{\mu}^{2}_1} \dots  \bar{\phi}_{\tilde{\mu}^{2}_{2r_{2}}}
		|0\rangle,   \label{o-right-state}
	\end{align}
	where
	\begin{align}
		l(\tilde{\mu}^{j})=
		\left\{\begin{array}{ll}
			2r_{j}-1, &  \tilde{\mu}^{2}_{2r_{2}}=0,
			\\  2r_{j},& \tilde{\mu}^{2}_{2r_{2}}>0,
		\end{array}
		\right. \quad j=1,2.
	\end{align}
	Note that
	$\tilde{\mu}^{j}=(\tilde{\mu}_{1}^{j},\tilde{\mu}_{2}^{j},\cdots,\tilde{\mu}_{2r_{j}}^{1})$ and
	$\tilde{\nu}^{j}=(\tilde{\nu}_{1}^{j},\tilde{\nu}_{2}^{j},\cdots,\tilde{\nu}_{2r_{j}}^{1})$,
	then
	\begin{align}\label{fock-scalar-product}
		\left( \langle \tilde{\mu}^2,\tilde{\mu}^1|,| \tilde{\nu}^{1},\tilde{\nu}^{2} \rangle \right)
		=&\langle 0|
		\bar{\phi}_{\tilde{\mu}^{2}_{2r_{2}}}^* \dots \bar{\phi}_{\tilde{\mu}^{2}_1}^*
		\phi_{\tilde{\mu}^{1}_{2r_{1}}}^* \dots \phi_{\tilde{\mu}^{1}_1}^*
		\phi_{\tilde{\mu}^{1}_1} \dots  \phi_{\tilde{\mu}^{1}_{2r_{1}}}
		\bar{\phi}_{\tilde{\mu}^{2}_1} \dots  \bar{\phi}_{\tilde{\mu}^{2}_{2r_{2}}}
		|0\rangle \notag\\
		=&2^{l(\tilde{\mu}^1)}2^{l(\tilde{\mu}^2)}
		\prod_{i=1}^{2r_{1}} \delta_{\tilde{\mu}_{i}^1,\tilde{\nu}_{i}^1}
		\prod_{j=1}^{2r_{2}} \delta_{\tilde{\mu}_{j}^2,\tilde{\nu}_{j}^2}
		= 2^{l(\tilde{\mu}^1)}2^{l(\tilde{\mu}^2)}\delta_{\tilde{\mu}^1,\tilde{\nu}^1}\delta_{\tilde{\mu}^2,\tilde{\nu}^2},
	\end{align}
	where
\begin{align}
\delta_{\tilde{\mu}^j,\tilde{\nu}^j}=
\left\{\begin{array}{ll}
	1, & \tilde{\mu}^j=\tilde{\nu}^j,
	\\  0,& otherwise,
\end{array}
\right. \quad j=1,2.		
\end{align}	
	
\section{The   generalized $i$-boson model}
In this section, we construct the   generalized $i$-boson model and present
the representation of the  generalized $i$-boson algebra. Meanwhile, the global intertwining equation are established.
	\subsection{The  generalized $i$-boson algebra and the space of states}
Define  the   generalized $i$-boson algebra based on the generators $\{\varphi^{(j)},\varphi^{(j)\dagger}, \mathcal{N}^{(j)}\}$ satisfying
\begin{align}
	[\varphi^{(i)},\varphi^{(j)\dagger}] =\delta_{i,j} (-1)^{\mathcal{N}^{(i)}}, \quad [\mathcal{N}^{(i)},\varphi^{(j)}] = -\delta_{i,j}\varphi^{(i)}, \quad [\mathcal{N}^{(i)},\varphi^{(j)\dagger}] =\delta_{i,j}\varphi^{(i)\dagger},\mathrm{~for~}i,j=1,2.
\end{align}
Introduce $M_{1}+M_{2}+2$ copies of the    generalized $i$-boson algebra generated by $\{\varphi_k^{(j)},\varphi_k^{(j)\dagger}, \mathcal{N}_k^{(j)}\}$, $0\leq k\leq M_j$.
Different copies and types of the    generalized $i$-boson algebra are assumed to commute, giving rise to  equations
\begin{align}\label{generalized-calculate}
	[\varphi_{l_{i}}^{(i)},\varphi_{k_{j}}^{(j)\dagger}]&=
	\left\{\begin{array}{ll}
		\delta_{l_{i},k_{j}}(-1)^{\mathcal{N}_{l_{i}}^{(i)}}, &  i=j,
		\\  0,&  i \neq j,
	\end{array}
	\right.	
	\quad	
	[\mathcal{N}_{l_{i}}^{(i)},\varphi_{k_{j}}^{(j)}]=
	\left\{\begin{array}{ll}
		-\delta_{l_{i},k_{j}}\varphi_{l_{i}}^{(i)}, &  i=j
		\\  0,&  i \neq j,
	\end{array}
	\right.
	\\
	[\mathcal{N}_{l_{i}}^{(i)},\varphi_{k_{j}}^{(j)\dagger}]&=
	\left\{\begin{array}{ll}
		\delta_{l_{i},k_{j}}\varphi_{l_{i}}^{(i)\dagger}, &  i=j,
		\\  0,&  i \neq j,
	\end{array}
	\right.
	\ \  \mathrm{~for~all~} i,j=1,2,  \mathrm{~and~}   0 \leq l_{i},k_{j} \leq M_j.
\end{align}

	Consider integral lattices composed of a pair of lattices
	of $M_1+1$ sites and $M_2+1$ sites. Construct this pair of lattices by placing the integer $n_i^{(j)}$ at the $i$-th position in the $j$-th lattice,  the occupation number $n_i^{(j)}$  satisfies
	\begin{align}
		0 \le n_0^{(j)},  ~    0 \le n_1^{(j)},\cdots,  n_{M_j}^{(j)} \le 1, \mathrm{~for~all~} j=1,2  \mathrm{~and~}   0 \leq i \leq M_j .
	\end{align}
	Define a space of states  $\tilde{\mathcal{V}}=\mathcal{V}^{(1)}\bigotimes	\tilde{\mathcal{V}}^{(2)}$
	as the linear span of all lattice configurations and denote the basis vector in $\tilde{\mathcal{V}}$ as
	\begin{align}
		\text{Basis}(\tilde{\mathcal{V}})
		=\left\{|\tilde{n}^{1}\rangle^{(1)}\bigotimes|\tilde{n}^{2}\rangle^{(2)}=	 \bigotimes_{i=0}^{M_{1}}|\tilde{n}^{1}_{i}\rangle_{i}^{(1)}\bigotimes\bigotimes_{j=0}^{M_{2}}|\tilde{n}^{2}_{j}\rangle_{j}^{(2)}  \right\},
	\end{align}
	where $\tilde{\mathcal{V}}^{(1)}$ and $\tilde{\mathcal{V}}^{(2)}$ are $M_1+1$  and $M_2+1$ copies of the  space of states respectively, and corresponding basis vectors have the following forms
	\begin{align}
		\text{Basis}	(\tilde{\mathcal{V}}^{(j)})&
		=\left\{|\tilde{n}^{j}\rangle^{(j)}=  |\tilde{n}_{0}^{j}\rangle_0^{(j)}\otimes\cdots\otimes|\tilde{n}_{M_j}^{j}\rangle_{M_j}^{(j)}                  =	 \bigotimes_{i=0}^{M_{j}}|\tilde{n}_{i}^{j} \rangle_{i}^{(j)} \right\},\quad j=1,2.
	\end{align}
	We write the number of particles as $N_j=\sum\limits_{i=0}^{M_j} n_i^j$.
	The inner product between two basis vectors $|\tilde{m}^{1}\rangle^{(1)}\bigotimes|\tilde{m}^{2}\rangle^{(2)}$ and $|\tilde{n}^{1}\rangle^{(1)}\bigotimes|\tilde{n}^{2}\rangle^{(2)}$ is given by
	\begin{align}\label{inner-product}
		\tilde{\mathcal{I}}
		\big(
		|\tilde{m}^{1} \rangle^{(1)}\bigotimes|\tilde{m}^{2}\rangle^{(2)}, |\tilde{n}^{1}\rangle^{(1)}\bigotimes|\tilde{n}^{2}\rangle^{(2)}
		\big)
		=\theta_0^{1}\theta_0^{2}\prod_{j=1}^{2}\prod_{i=1}^{M_j} 2^{-\tilde{m}_i^{j}}\delta_{\tilde{m}_i^j, \tilde{n}_i^j},
	\end{align}
	where
	\begin{align}
		\theta_0^{j}= 2^{\tilde{m}_0^{j}}\theta\left(0 \leq m_0^{j} \leq 1\right)\delta_{\tilde{m}_0^{j},\tilde{n}_0^{j}}, \quad
		\theta(x)=
		\left\{\begin{array}{ll}
			1, &  x~\mathrm{is~true},
			\\  0,&  x~\mathrm{is~false}.
		\end{array}
		\right.	
	\end{align}
	
	The dual space of states  $\tilde{\mathcal{V}}^{\ast}=\tilde{\mathcal{V}}^{(1)\ast}\bigotimes\tilde{\mathcal{V}}^{(2)\ast} $ has the basis vector
	\begin{align}
		\text{Basis}	(\tilde{\mathcal{V}}^{\ast})
		=\left\{
		\prescript{(2)}{}{\langle \tilde{m}^{2}|}\bigotimes \prescript{(1)}{}{\langle \tilde{m}^{1}|}=
		\bigotimes_{j=0}^{M_{2}}\prescript{(2)}{j}{\langle \tilde{m}^{2}_j|}\bigotimes
		\bigotimes_{i=0}^{M_{1}}\prescript{(1)}{i}{\langle \tilde{m}^{1}_i|}
		\right\},
	\end{align}
	where
	\begin{align}
		\text{Basis}	(\tilde{\mathcal{V}}^{(j){\ast}})&
		=\left\{  \prescript{(j)}{}{\langle \tilde{m}^{j}|}=  \bigotimes_{i=0}^{M_{j}}\prescript{(j)}{i}{\langle \tilde{m}^{j}_i|} \right\}, \quad j=1,2,		
	\end{align}
	and the integer $m_i^{(j)}$ satisfies
	\begin{align}
		0 \le \tilde{m}_0^{(j)},  ~    0 \le \tilde{m}_1^{(j)},\cdots,  \tilde{m}_{M_j}^{(j)} \le 1, \mathrm{~for~all~} j=1,2 \mathrm{~and~}   0 \leq i \leq M_j .
	\end{align}
	Meanwhile,
	the action of the basis vector $\prescript{(2)}{}{\langle \tilde{m}^{2}|}\bigotimes \prescript{(1)}{}{\langle \tilde{m}^{1}|} \in \tilde{\mathcal{V}}^{\ast}$ on $|\tilde{n}^{1}\rangle^{(1)}\bigotimes|\tilde{n}^{2}\rangle^{(2)} \in \tilde{\mathcal{V}}$ is as follows
	\begin{align}\label{tgibm-scalar-product}
		\prescript{(2)}{}{\langle \tilde{m}^{2}|}\bigotimes \prescript{(1)}{}{\langle \tilde{m}^{1}|}
		\tilde{n}^{1}\rangle^{(1)}\bigotimes|\tilde{n}^{2}\rangle^{(2)}
		=\tilde{\mathcal{I}}
		\big(
		|\tilde{m}^{1} \rangle^{(1)}\bigotimes|\tilde{m}^{2}\rangle^{(2)}, |\tilde{n}^{1}\rangle^{(1)}\bigotimes|\tilde{n}^{2}\rangle^{(2)}
		\big).
	\end{align}

\subsection{The representation of the   generalized $i$-boson algebra}
The action of $\{\varphi_i^{(1)},\varphi_i^{(1)\dagger}, \mathcal{N}_i^{(1)} \}$ $( 1\leq i\leq M_1)$ on the basis vector $\bigotimes\limits_{i=0}^{M_{1}}|\tilde{n}_{i}^{1} \rangle_{i}^{(1)}\bigotimes\bigotimes\limits_{j=0}^{M_{2}}|\tilde{n}_{j}^{2}\rangle_{j}^{(2)} $ in the vector space $\tilde{\mathcal{V}}$ satisfies
	\begin{align}\label{right-ym}
		&\varphi_i^{(1)} \bigotimes_{i=0}^{M_{1}}|\tilde{n}_{i}^{1} \rangle_{i}^{(1)}\bigotimes\bigotimes_{j=0}^{M_{2}}|\tilde{n}_{j}^{2}\rangle_{j}^{(2)}   \notag\\
		=&\left\{
		\begin{array}
			{ll}0, & \tilde{n}_{i}^{1}=0, \\
			\\
			\frac{1}{\sqrt{2}}|\tilde{n}_{0}^{1} \rangle_0^{(1)}\otimes\cdots\otimes|\tilde{n}_{i}^{1}-1\rangle_i^{(1)}\otimes\cdots\otimes|\tilde{n}_{M_1}^{1}\rangle_{M_1}^{(1)}\bigotimes\bigotimes\limits_{j=0}^{M_{2}}|\tilde{n}_{j}^{2}\rangle_{j}^{(2)}, & \tilde{n}_{i}^{1}\geq1,
		\end{array}\right.
	\end{align}	
	where $\varphi_i^{(1)}$   as an annihilator will  reduce particles from the $i$-th lattice site.
Creation operators $\varphi_i^{(1)\dagger}$   increase the $i$-th occupation number of the basis vector by $1$ as follows
\begin{align}\label{right-sc}
&\varphi_i^{(1)\dagger}
\bigotimes_{i=0}^{M_{1}}|\tilde{n}_{i}^{1} \rangle_{i}^{(1)}\bigotimes\bigotimes_{j=0}^{M_{2}}|\tilde{n}_{j}^{2}\rangle_{j}^{(2)} 	\notag\\	
=&\frac{1 + (-1)^{\tilde{n}_i^{1}}}{\sqrt{2}} |\tilde{n}_{0}^{1}\rangle_0^{(1)}\otimes\cdots\otimes|\tilde{n}_{i}^{1}+1\rangle_i^{(1)}\otimes\cdots\otimes|\tilde{n}_{M_1}^{1}\rangle_{M_1}^{(1)}\bigotimes\bigotimes\limits_{j=0}^{M_{2}}|\tilde{n}_{j}^{2}\rangle_{j}^{(2)}.
\end{align}	
However, for $i=0$,
\begin{align}
&\varphi_0^{(1)}
\bigotimes_{i=0}^{M_{1}}|\tilde{n}_{i}^{1} \rangle_{i}^{(1)}\bigotimes\bigotimes_{j=0}^{M_{2}}|\tilde{n}_{j}^{2}\rangle_{j}^{(2)} 		=\frac{1 + (-1)^{\tilde{n}_0^{1}}}{\sqrt{2}} |\tilde{n}_{0}^{1}-1\rangle_0^{(1)}\otimes\cdots\otimes|\tilde{n}_{M_1}^{1}\rangle_{M_1}^{(1)}\bigotimes\bigotimes\limits_{j=0}^{M_{2}}|\tilde{n}_{j}^{2}\rangle_{j}^{(2)}, \label{right-ym0}\\
		&\varphi_0^{(1)\dagger}
		\bigotimes_{i=0}^{M_{1}}|\tilde{n}_{i}^{1} \rangle_{i}^{(1)}\bigotimes\bigotimes_{j=0}^{M_{2}}|\tilde{n}_{j}^{2}\rangle_{j}^{(2)} 	=\frac{1}{\sqrt{2}} |\tilde{n}_{0}^{1}+1\rangle_0^{(1)}\otimes\cdots\otimes|\tilde{n}_{M_1}^{1}\rangle_{M_1}^{(1)}\bigotimes\bigotimes\limits_{j=0}^{M_{2}}|\tilde{n}_{j}^{2}\rangle_{j}^{(2)}.\label{right-sc0}	 \end{align}	
	The action of $\mathcal{N}_i^{(1)}$ $(0\leq i\leq M_j)$ is defined as
	\begin{align}\label{right-eigenvalue}
		\mathcal{N}_i^{(1)}
		\bigotimes_{i=0}^{M_{1}}|n_{i}^{1} \rangle_{i}^{(1)}\bigotimes\bigotimes_{j=0}^{M_{2}}|n_{j}^{2}\rangle_{j}^{(2)}  	
		=n_{i}^{1}
		\bigotimes_{i=0}^{M_{1}}|n_{i}^{1} \rangle_{i}^{(1)}\bigotimes\bigotimes_{j=0}^{M_{2}}|n_{j}^{2}\rangle_{j}^{(2)} .
	\end{align}	
	Meanwhile,  $\{\varphi_i^{(1)},\varphi_i^{(1)\dagger}, \mathcal{N}_i^{(1)} \}$ $( 1\leq i\leq M_1)$  act on the the basis vector in the  $\tilde{\mathcal{V}}^{\ast}$  by
	\begin{align}
		&\bigotimes_{j=0}^{M_{2}}\prescript{(2)}{j}{\langle \tilde{m}^{2}_j|}\bigotimes
		\bigotimes_{i=0}^{M_{1}}\prescript{(1)}{i}{\langle \tilde{m}^{1}_i|}
		\varphi_i^{(1)\dagger}\notag\\
		=&\frac{1}{\sqrt{2}}  \delta_{\tilde{m}^{1}_i -1 \geq0}
		\bigotimes\limits_{j=0}^{M_{2}}\prescript{(2)}{j}{\langle \tilde{m}^{2}_j|}\bigotimes	
		\prescript{(1)}{0}{\langle \tilde{m}^{1}_0|}\otimes\cdots\otimes
		\prescript{(1)}{i}{\langle \tilde{m}^{1}_{i}-1|}
		\otimes\cdots\otimes\prescript{(1)}{M_1}{\langle \tilde{m}^{1}_{M_1}|},~~1\leq i\leq M_1,\label{left-ym}\\
		&\bigotimes_{j=0}^{M_{2}}\prescript{(2)}{j}{\langle \tilde{m}^{2}_j|}\bigotimes
		\bigotimes_{i=0}^{M_{1}}\prescript{(1)}{i}{\langle \tilde{m}^{1}_i|}
		\varphi_i^{(1)}\notag\\
		=&	\frac{1 + (-1)^{\tilde{m}_i^{1}}}{\sqrt{2}}
		\bigotimes\limits_{j=0}^{M_{2}}\prescript{(2)}{j}{\langle \tilde{m}^{2}_j|}\bigotimes	
		\prescript{(1)}{0}{\langle \tilde{m}^{1}_0|}\otimes\cdots\otimes
		\prescript{(1)}{i}{\langle \tilde{m}^{1}_{i}+1|}
		\otimes\cdots\otimes\prescript{(1)}{M_1}{\langle \tilde{m}^{1}_{M_1}|},~~1\leq i\leq M_1,\label{left-sc}\\
		&\bigotimes_{j=0}^{M_{2}}\prescript{(2)}{j}{\langle \tilde{m}^{2}_j|}\bigotimes
		\bigotimes_{i=0}^{M_{1}}\prescript{(1)}{i}{\langle \tilde{m}^{1}_i|}
		\varphi_0^{(1)}
		=\frac{1}{\sqrt{2}}
		\bigotimes\limits_{j=0}^{M_{2}}\prescript{(2)}{j}{\langle \tilde{m}^{2}_j|}\bigotimes
		\prescript{(1)}{0}{\langle \tilde{m}^{1}_0 +1|}\otimes\cdots\otimes\prescript{(1)}{M_1}{\langle \tilde{m}^{1}_{M_1}|},\label{left-sc0} \\
		&\bigotimes_{j=0}^{M_{2}}\prescript{(2)}{j}{\langle \tilde{m}^{2}_j|}\bigotimes
		\bigotimes_{i=0}^{M_{1}}\prescript{(1)}{i}{\langle \tilde{m}^{1}_i|}
		\varphi_0^{(1)\dagger}
		=	\frac{1 - (-1)^{\tilde{m}_0^{1}}}{\sqrt{2}}
		\bigotimes\limits_{j=0}^{M_{2}}\prescript{(2)}{j}{\langle \tilde{m}^{2}_j|}\bigotimes
		\prescript{(1)}{0}{\langle \tilde{m}^{1}_0 -1|}\otimes\cdots\otimes\prescript{(1)}{M_1}{\langle \tilde{m}^{1}_{M_1}|},\label{left-ym0}\\
		&\bigotimes_{j=0}^{M_{2}}\prescript{(2)}{j}{\langle \tilde{m}^{2}_j|}\bigotimes
		\bigotimes_{i=0}^{M_{1}}\prescript{(1)}{i}{\langle \tilde{m}^{1}_i|}
		\mathcal{N}_i^{(1)}
		=\tilde{m}^{1}_i
		\bigotimes\limits_{j=0}^{M_{2}}\prescript{(2)}{j}{\langle \tilde{m}^{2}_j|}\bigotimes
		\prescript{(1)}{0}{\langle \tilde{m}^{1}_0|}\otimes\cdots\otimes\prescript{(1)}{M_1}{\langle \tilde{m}^{1}_{M_1}|},~~ 0\leq i\leq M_1,\label{left-eigenvalue}
	\end{align}
	where 	$ \delta_{\tilde{m}^{1}_i -1 \geq0}=1$ if $\tilde{m}^{1}_i -1 \geq0$.
	The actions of $\{\varphi^{(2)}_{j},\varphi^{(2)\dagger}_{j}, \mathcal{N}^{(2)}_{j} \}$ $( 1\leq j\leq M_2)$ on the basis vectors are similar to  $\{\varphi^{(1)}_{i},\varphi^{(1)\dagger}_{i}, \mathcal{N}^{(1)}_{i}\}$ and it will not be repeated.

	\subsection{The $R$-matrix and the global intertwining equation}
	Introduce the  $R$-matrix for the   generalized $i$-boson  model
	\begin{align}
		\tilde{R}(x, y) =
		\begin{pmatrix}
			x + y & 0 & 0 & 0 \\
			0 & y - x & 2x^{\frac{1}{2}}y^{\frac{1}{2}} & 0 \\
			0 & 2x^{\frac{1}{2}}y^{\frac{1}{2}} & x - y & 0 \\
			0 & 0 & 0 & x + y
		\end{pmatrix}.
	\end{align}
	Define the $L$-matrix
	\begin{align}\label{L-j}
		\tilde{L}_{i}^{(j)}(x)=
		\begin{pmatrix}
			x^{-\frac{1}{2}} &  \sqrt{2}\varphi_i^{(j)\dagger} \\
			\sqrt{2}\varphi_i^{(j)} & x^{\frac{1}{2}}
		\end{pmatrix}, \quad j=1,2\mathrm{~and~}i=0,\cdots,M_j,
	\end{align}
	and the monodromy matrix		
	\begin{align}
		\tilde{T}(x)=\tilde{T}_{2}(x) \cdot \tilde{T}_{1}(x),
	\end{align}
	where	
	\begin{align}\label{T-j}
		\begin{aligned}
			\tilde{T}_{j}(x) =\tilde{L}_{M_{j}}^{(j)}(x)\tilde{L}_{M_{j}-1}^{(j)}(x)\cdots \tilde{L}_{0}^{(j)}(x)
			=
			\begin{pmatrix}
				\tilde{A}_j(x) & \tilde{B}_j(x) \\
				\tilde{C}_j(x) & \tilde{D}_j(x)
			\end{pmatrix}, \mathrm{~for~} j=1,2.
		\end{aligned}
	\end{align}
	The monodromy matrix satisfies the global intertwining equation
	\begin{align}
		\tilde{R}(x,y)(\tilde{T}(x)\otimes \tilde{T}(y))&=(\tilde{T}(y)\otimes \tilde{T}(x))\tilde{R}(x,y),
	\end{align}
	and the following bilinear equations  hold
	\begin{align}
		\tilde{R}(x,y)(\tilde{T}_{j}(x)\otimes \tilde{T}_{j}(y))&=(\tilde{T}_{j}(y)\otimes \tilde{T}_{j}(x))\tilde{R}(x,y),\\
		\tilde{R}(x,y)(\tilde{L}_{i}^{(j)}(x)\otimes \tilde{L}_{i}^{(j)}(y))&=(\tilde{L}_{i}^{(j)}(y)\otimes \tilde{L}_{i}^{(j)}(x))\tilde{R}(x,y). \label{bilinear-3}
	\end{align}
	Particularly,
	\begin{align}\label{B,C,commute}
		[\tilde{B}_i(x),\tilde{B}_j(y)]=[\tilde{C}_i(x),\tilde{C}_j(y)]=0, \mathrm{~for~} i,j=1,2.
	\end{align}
	
\section{The actions of the monodromy matrix operators on basis vectors}
	In this section, we establish the relation between  strict boxed   $2$-partitions and  basis vectors by  maps $\mathcal{M}$ and $\mathcal{M}^*$.
	Moreover, the action of  monodromy matrix operators on basis vectors has been investigated, which generates  interlacing  strict boxed   $2$-partitions.
	
	For basis vectors
	$| \tilde{m}^{1}\rangle^{(1)}\bigotimes|\tilde{m}^{2}\rangle^{(2)}$,
	$| \tilde{n}^{1}\rangle^{(1)}\bigotimes|\tilde{n}^{2}\rangle^{(2)}$ in $\tilde{\mathcal{V}}$ and
	$\prescript{(2)}{}{\langle \tilde{m}^{2}|}\bigotimes \prescript{(1)}{}{\langle \tilde{m}^{1}|}$,
	$\prescript{(2)}{}{\langle \tilde{n}^{2}|}\bigotimes \prescript{(1)}{}{\langle \tilde{n}^{1}|}$ in  $\tilde{\mathcal{V}}^*$, set
	\begin{align}
		\Sigma_i^{\tilde{m}^j} = \sum_{k=i}^{M_j} \tilde{m}_k^j, \quad
		\Sigma_i^{\tilde{n}^j} = \sum_{l=i}^{M_j} \tilde{n}_l^j,\quad j=1,2  \text{~and~} 0\leq i \leq M_{j}.
	\end{align}
	The symbol $| \tilde{m}^{j}\rangle^{(j)} \triangleright   |   \tilde{n}^{j}\rangle^{(j)}$
	denotes that $| \tilde{m}^{j}\rangle^{(j)}$ is admissible to   $ | \tilde{n}^{j}\rangle^{(j)}$, if and only if
	\begin{align}\label{xrxtj}
		(\Sigma_0^{\tilde{m}^j} - \Sigma_0^{\tilde{n}^j}) = 1\quad \text{and} \quad 0 \leq (\Sigma_i^{\tilde{m}^j} - \Sigma_i^{\tilde{n}^j}) \leq 1.
	\end{align}
	Define  $| \tilde{m}^{1}\rangle^{(1)}\bigotimes|\tilde{m}^{2}\rangle^{(2)}$  is admissible to $| \tilde{n}^{1}\rangle^{(1)}\bigotimes|\tilde{n}^{2}\rangle^{(2)}$ giving by
	\begin{align}\label{admiss}
		| \tilde{m}^{1}\rangle^{(1)}\bigotimes|\tilde{m}^{2}\rangle^{(2)} \triangleright  | \tilde{n}^{1}\rangle^{(1)}\bigotimes|\tilde{n}^{2}\rangle^{(2)}
		\Longleftrightarrow
		| \tilde{m}^{1}\rangle^{(1)} \triangleright   | \tilde{n}^{1}\rangle^{(1)} \text{~and~}
		| \tilde{m}^{2}\rangle^{(2)} \triangleright   |   \tilde{n}^{2}\rangle^{(2)}	.
	\end{align}

	Meanwhile, $\prescript{(j)}{}{\langle \tilde{n}^{j}|}  \triangleleft \prescript{(j)}{}{\langle \tilde{m}^{j}|} $  follows that
	\begin{align}\label{admiss-left}
		\prescript{(2)}{}{\langle \tilde{n}^{2}|}\bigotimes\prescript{(1)}{}{\langle \tilde{n}^{1}|}
		\triangleleft
		\prescript{(2)}{}{\langle \tilde{m}^{2}|}\bigotimes\prescript{(1)}{}{\langle \tilde{m}^{1}|}
		\Longleftrightarrow
		\prescript{(1)}{}{\langle \tilde{n}^{1}|}  \triangleleft \prescript{(1)}{}{\langle \tilde{m}^{1}|}
		\text{~and~}
		\prescript{(2)}{}{\langle \tilde{n}^{2}|}  \triangleleft \prescript{(2)}{}{\langle \tilde{m}^{2}|}.
	\end{align}
	
	\begin{definition}\label{M}
		Define  linear maps
		\begin{align}
			\tilde{\mathcal{M}} : \tilde{\mathcal{V}} \to \tilde{\mathcal{F}}_{0,0}, \quad
			\tilde{\mathcal{M}}^* : \tilde{\mathcal{V}}^* \to \tilde{\mathcal{F}}_{0,0}^{*},
		\end{align}
		satisfying
		\begin{align}\label{M-map}
			\tilde{\mathcal{M}}  | \tilde{n}^{1}\rangle^{(1)}\bigotimes|\tilde{n}^{2}\rangle^{(2)}&=
			2^{-l(\tilde{\nu}^{1})-l(\tilde{\nu}^{2})}	
			| \tilde{\nu}^{1},\tilde{\nu}^{2} \rangle =
			2^{-l(\tilde{\nu}^{1})-l(\tilde{\nu}^{2})}	
			| \tilde{\chi} \rangle,\\
			\prescript{(2)}{}{\langle \tilde{n}^{2}|}\bigotimes \prescript{(1)}{}{\langle \tilde{n}^{1}|}  \tilde{\mathcal{M}}^* &=
			2^{-l(\tilde{\nu}^{1})-l(\tilde{\nu}^{2})}	
			\langle \tilde{\nu}^{2},\tilde{\nu}^{1}|
			=2^{-l(\tilde{\nu}^{1})-l(\tilde{\nu}^{2})}	
			\langle \tilde{\chi}|,
		\end{align}
		where the strict partition
		\begin{align}
			\tilde{\nu}^{j}=(1^{\tilde{n}_1^{j}} 2^{\tilde{n}_2^{j}} \ldots M_{j}^{\tilde{n}_{M_{j}}^{j}} ), \quad j=1,2 .
		\end{align}	
	\end{definition}
	It can be inferred from the Lemma $4.2$ in Reference \cite{t-g-phase-model-boxed-UC} that the following lemma holds.
	\begin{lemma}\label{xr-jc-1}
		Assume that   maps   $\tilde{\mathcal{M}} $ and $\tilde{\mathcal{M}}^*$  satisfy
		\begin{align}\label{M-map}
			\tilde{\mathcal{M}}  | \tilde{m}^{1}\rangle^{(1)}\bigotimes|\tilde{m}^{2}\rangle^{(2)}&=
			2^{-l(\tilde{\mu}^{1})-l(\tilde{\mu}^{2})}	
			| \tilde{\mu}^{1},\tilde{\mu}^{2} \rangle , 	\\
			\tilde{\mathcal{M}}  | \tilde{n}^{1}\rangle^{(1)}\bigotimes|\tilde{n}^{2}\rangle^{(2)}&=
			2^{-l(\tilde{\nu}^{1})-l(\tilde{\nu}^{2})}	
			| \tilde{\nu}^{1},\tilde{\nu}^{2} \rangle, \\
			\prescript{(2)}{}{\langle \tilde{m}^{2}|}\bigotimes \prescript{(1)}{}{\langle \tilde{m}^{1}|}  \tilde{\mathcal{M}}^* &=
			2^{-l(\tilde{\mu}^{1})-l(\tilde{\mu}^{2})}	
			\langle \tilde{\mu}^{2},\tilde{\mu}^{1}|, \\	
			\prescript{(2)}{}{\langle \tilde{n}^{2}|}\bigotimes \prescript{(1)}{}{\langle \tilde{n}^{1}|}  \tilde{\mathcal{M}}^* &=
			2^{-l(\tilde{\nu}^{1})-l(\tilde{\nu}^{2})}	
			\langle \tilde{\nu}^{2},\tilde{\nu}^{1}| ,
		\end{align}	
		then
		\begin{align}
			&| \tilde{m}^{1}\rangle^{(1)}\bigotimes|\tilde{m}^{2}\rangle^{(2)} \triangleright  | \tilde{n}^{1}\rangle^{(1)}\bigotimes|\tilde{n}^{2}\rangle^{(2)}
			\implies
			\tilde{\mu}^{1}  \succ  \tilde{\nu}^{1} ,\tilde{\mu}^{2}  \succ  \tilde{\nu}^{2} , \label{r-admiss-interlace}\\
			&\prescript{(2)}{}{\langle \tilde{n}^{2}|}\bigotimes \prescript{(1)}{}{\langle \tilde{n}^{1}|}
			\triangleleft
			\prescript{(2)}{}{\langle \tilde{m}^{2}|}\bigotimes \prescript{(1)}{}{\langle \tilde{m}^{1}|}
			\implies  \tilde{\nu}^{1}  \prec  \tilde{\mu}^{1} , \tilde{\nu}^{2} \prec  \tilde{\mu}^{2}.	 \label{l-admiss-interlace}
		\end{align}	
	\end{lemma}

	\begin{lemma}\label{B-n,C-n,2}
		Set monodromy matrix operators
		$\tilde{\mathbb{B}}_{j}(x) = x^{\frac{M_{j}}{2}} \tilde{B}_{j}(x)$ and $\tilde{\mathbb{C}}_{j}(x) = x^{\frac{M_{j}}{2}} \tilde{C}(\frac{1}{x})$, $j=1,2$.
		For basis vectors
		$| \tilde{n}^{1}\rangle^{(1)}\bigotimes|\tilde{n}^{2}\rangle^{(2)} \in \tilde{\mathcal{V}}$ and
		$\prescript{(2)}{}{\langle \tilde{n}^{2}|}\bigotimes \prescript{(1)}{}{\langle \tilde{n}^{1}|} \in \tilde{\mathcal{V}}^*$,
		the actions of  monodromy matrix operators on basis vectors are given by
		\begin{align}
			 &\tilde{\mathbb{B}}_{1}(x)\tilde{\mathbb{B}}_{2}(y)|\tilde{n}^{1}\rangle^{(1)}\bigotimes|\tilde{n}^{2}\rangle^{(2)} \notag\\
			=&
			\sum_{{ |\tilde{m}^{1} \rangle^{(1)} \triangleright | \tilde{n}^{1} \rangle^{(1)} }\atop{|\tilde{m}^{2} \rangle^{(2)} \triangleright | \tilde{n}^{2} \rangle^{(2)} }}
			\prod_{i=1}^{M_{1}} 2^{\delta_{(\tilde{m}^{1}_i - \tilde{n}^{1}_i), 1}}
			x^{i(\tilde{m}^{1}_i - \tilde{n}^{1}_i)}
			\prod_{j=1}^{M_{2}} 2^{\delta_{(\tilde{m}^{2}_j - \tilde{n}^{2}_j), 1}}
			y^{j(\tilde{m}^{2}_j - \tilde{n}^{2}_j)}
			| \tilde{m}^{1}\rangle^{(1)}\bigotimes|\tilde{m}^{2}\rangle^{(2)}, \label{1-B-n,C-n,2}  \\
			&\prescript{(2)}{}{\langle \tilde{n}^{2}|}\bigotimes \prescript{(1)}{}{\langle \tilde{n}^{1}|}
			\tilde{\mathbb{C}}_{2}(y)\tilde{\mathbb{C}}_{1}(x)\notag\\
			=&
			\sum_{{ \prescript{(1)}{}{\langle \tilde{n}^{1}|}
					\triangleleft
					\prescript{(1)}{}{\langle \tilde{m}^{1}|} }
				\atop{\prescript{(2)}{}{\langle \tilde{n}^{2}|}
					\triangleleft
					\prescript{(2)}{}{\langle \tilde{m}^{2}|}  }}
			\prod_{i=1}^{M_{1}} 2^{\delta_{(\tilde{m}^{1}_i - \tilde{n}^{1}_i), 1}}
			x^{i(\tilde{m}^{1}_i - \tilde{n}^{1}_i)}
			\prod_{j=1}^{M_{2}} 2^{\delta_{(\tilde{m}^{2}_j - \tilde{n}^{2}_j), 1}}
			y^{j(\tilde{m}^{2}_j - \tilde{n}^{2}_j)}
			\prescript{(2)}{}{\langle \tilde{m}^{2}|}\bigotimes \prescript{(1)}{}{\langle \tilde{m}^{1}|} .\label{2-B-n,C-n,2}
		\end{align}
	\end{lemma}
	\begin{proof}
		 Eq.$(\ref{T-j})$ leads to
		\begin{align}
			\tilde{B}_1 (x) \tilde{B}_2 (y)=&( 1\ 0 )
			\begin{pmatrix}
				\tilde{A}_1(x) & \tilde{B}_1(x) \\
				\tilde{C}_1(x) & \tilde{D}_1(x)
			\end{pmatrix}
			\begin{pmatrix}
				0 \\
				1
			\end{pmatrix}
			( 1\ 0 )
			\begin{pmatrix}
				\tilde{A}_1(x) & \tilde{B}_1(x) \\
				\tilde{C}_1(x) & \tilde{D}_1(x)
			\end{pmatrix}
			\begin{pmatrix}
				0 \\
				1
			\end{pmatrix}\notag\\
			=& \uparrow^*  \tilde{L}_{M_1}^{(1)}(x)\cdots \tilde{L}_0^{(1)}(x)\downarrow \uparrow^* \tilde{L}_{M_2}^{(2)}(y)\cdots \tilde{L}_0^{(2)}(y) \downarrow,
		\end{align}	
		where
		\begin{align}
			\begin{pmatrix}
				1 \\
				0
			\end{pmatrix}= \uparrow, \quad
			\begin{pmatrix}
				0 \\
				1
			\end{pmatrix}= \downarrow, \quad ( 1\ 0 )=\uparrow^*, \quad
			( 0\ 1 )=\downarrow^*.
		\end{align}
		Then
		\begin{align}\label{B1-B2}
			&\prescript{(2)}{}{\langle \tilde{m}^{2}|}\bigotimes \prescript{(1)}{}{\langle \tilde{m}^{1}|}
			\tilde{B}_1 (x) \tilde{B}_2 (y)
			|\tilde{n}^{1}\rangle^{(1)}\bigotimes|\tilde{n}^{2}\rangle^{(2)}\notag\\
			=&\prescript{(1)}{}{\langle \tilde{m}^{1}|}
			\tilde{B}_1 (x)
			|\tilde{n}^{1}\rangle^{(1)}
			\prescript{(2)}{}{\langle \tilde{m}^{2}|}
			\tilde{B}_2 (y)
			|\tilde{n}^{2}\rangle^{(2)}\notag\\
			=&\uparrow^*
			\mathbb{L}_{M_1}^{(1)}(x)\cdots \mathbb{L}_0^{(1)}(x)\downarrow \uparrow^* \mathbb{L}_{M_2}^{(2)}(y)\cdots \mathbb{L}_0^{(2)}(y)
			\downarrow,
		\end{align}
		where
		\begin{align}
			\mathbb{L}_{i}^{(j)}(x) =
			\begin{pmatrix}
				\prescript{(j)}{i}{\langle \tilde{m}^{j}_i|}  x^{- \frac{1}{2}} | \tilde{n}_i^{j} \rangle_i^{(j)} & \sqrt{2}\prescript{(j)}{i}{\langle \tilde{m}^{j}_i|}\varphi_i^{(j)\dagger} | \tilde{n}_i^{j} \rangle_i^{(j)} \\
				\sqrt{2}\prescript{(j)}{i}{\langle \tilde{m}^{j}_i|}\varphi_i^{(j)} | \tilde{n}_i^{j} \rangle_i^{(j)} &\prescript{(j)}{i}{\langle \tilde{m}^{j}_i|} x^{\frac{1}{2}}
				| \tilde{n}_i^{j} \rangle_i^{(j)}
			\end{pmatrix} , \quad j=1,2  \text{~and~} 0\leq i \leq M_{j}.
		\end{align}
		According to Eq.$(\ref{tgibm-scalar-product})$ and Eqs.$(\ref{right-ym})$-$(\ref{left-ym0})$,  we have
		\begin{align}
			\mathbb{L}_{i}^{(j)}(x) =
			\begin{cases}
				2^{-\tilde{m}_{i}^{j}}
				\begin{pmatrix}
					x^{-\frac{1}{2}} & 0 \\
					0 & x^{\frac{1}{2}}
				\end{pmatrix}, & \tilde{m}_i^{j} = \tilde{n}_i^{j}, \\
				2^{-\tilde{m}_{i}^{j}} (1 + (-1)^{\tilde{m}_i^{j}-1})
				\begin{pmatrix}
					0 & 1 \\
					0 & 0
				\end{pmatrix}, & \tilde{m}^{j}_i = \tilde{n}_i^{j} + 1, \\
				2^{-\tilde{m}_{i}^{j}}
				\begin{pmatrix}
					0 & 0 \\
					1 & 0
				\end{pmatrix}, & \tilde{m}_i^{j} = \tilde{n}_i^{j} - 1, \\
				2^{-\tilde{m}_{i}^{j}}
				\begin{pmatrix}
					0 & 0 \\
					0 & 0
				\end{pmatrix}, & \text{otherwise},
			\end{cases}
		\end{align}
		for all $j=1,2$ and $1\leq i \leq M_{j}$, as well as
		\begin{align}
			\mathbb{L}_{0}^{(j)}(x) =
			\begin{cases}
				\theta_0^{j}
				\begin{pmatrix}
					x^{-\frac{1}{2}} & 0 \\
					0 & x^{\frac{1}{2}}
				\end{pmatrix}, & \tilde{m}_0^{j} = \tilde{n}_0^{j}, \\
				\theta_0^{j}
				\begin{pmatrix}
					0 & 1 \\
					0 & 0
				\end{pmatrix}, & \tilde{m}^{j}_0 = \tilde{n}_0^{j} + 1, \\
				\theta_0^{j}(1 - (-1)^{\tilde{m}_0^{j}+1})
				\begin{pmatrix}
					0 & 0 \\
					1 & 0
				\end{pmatrix}, & \tilde{m}_0^{j} = \tilde{n}_0^{j} - 1, \\
				\theta_0^{j}
				\begin{pmatrix}
					0 & 0 \\
					0 & 0
				\end{pmatrix}, & \text{otherwise}.
			\end{cases}
		\end{align}
		Setting
		\begin{align}
			E_{12}=\begin{pmatrix}
				0 & 1 \\
				0 & 0
			\end{pmatrix}, \quad
			E_{21}=\begin{pmatrix}
				0 & 0 \\
				1 & 0
			\end{pmatrix}.
		\end{align}
		When $|\tilde{m}^{j}\rangle^{(j)} \not\triangleright |\tilde{n}^j \rangle^{(j)} $, $\tilde{m}^{j}_i = \tilde{n}_i^{j} + 1$ or $\tilde{m}_i^{j} = \tilde{n}_i^{j}-1$ will occur consecutively.
		Since $E_{12}^2=E_{21}^2=0$, then
		\begin{align}
			\uparrow^*
			\mathbb{L}_{M_1}^{(1)}(x)\cdots \mathbb{L}_0^{(1)}(x)\downarrow \uparrow^* \mathbb{L}_{M_2}^{(2)}(y)\cdots \mathbb{L}_0^{(2)}(y)
			\downarrow=0.
		\end{align}
		If $|\tilde{m}^{j}\rangle^{(j)} \triangleright |n^j \rangle^{(j)}$, $\tilde{m}^{j}_i = \tilde{n}_i^{j} + 1$ and $\tilde{m}_i^{j} = \tilde{n}_i^{j}-1$ appear alternately.
		Let the sets consisting of the subscripts satisfying $\tilde{m}^{j}_i = \tilde{n}_i^{j} + 1$ and $\tilde{m}_i^{j} = \tilde{n}_i^{j}-1$, which be denoted as  $\{p_1^{j},\cdots,p_{r_j}^{j} \}$ and $\{q_1^{j},\cdots,q_{s_j}^{j} \}$, respectively.
		Meanwhile, it follows from Eq.(\ref{xrxtj}) that the subscripts satisfy  $r_j -1 = s_j$ and
		\begin{align}
			p_i^{j} < q_i^{j} < p_{i+1}^{j}, \quad \text{for all } j=1,2, \text{ and }1 \leq i \leq r_j - 1.
		\end{align}
	Eq.$(\ref{B1-B2})$ is rewritten as
		\begin{align}
			&\prescript{(2)}{}{\langle \tilde{m}^{2}|}\bigotimes \prescript{(1)}{}{\langle \tilde{m}^{1}|}
			\tilde{B}_1 (x) \tilde{B}_2 (y)
			|\tilde{n}^{1}\rangle^{(1)}\bigotimes|\tilde{n}^{2}\rangle^{(2)}\notag\\
			=&\uparrow^*
			\mathbb{L}_{M_1}^{(1)}(x)\cdots \mathbb{L}_0^{(1)}(x)A_{21}\uparrow
			\uparrow^* \mathbb{L}_{M_2}^{(2)}(y)\cdots \mathbb{L}_0^{(2)}(y) A_{21}
			\uparrow    \notag\\
			=&\frac{\theta_0^{1}}{\prod\limits_{i=1}^{M_1} 2^{\tilde{m}_{i}^{1}}}
			\uparrow^* \rprod_{i=1}^{r_1}
			\left[
			\begin{pmatrix}
				x^{-\frac{1}{2}} & 0 \\
				0 & x^{\frac{1}{2}}
			\end{pmatrix}^{q_i^1 - p_{i-1}^1 - 1}
			\begin{pmatrix}
				0 & (1 + (-1)^{\tilde{m}_i^{1}-1}) \\
				0 & 0
			\end{pmatrix}
			\begin{pmatrix}
				x^{-\frac{1}{2}} & 0 \\
				0 & x^{\frac{1}{2}}
			\end{pmatrix}^{p_i^1 - q_{i-1}^1 - 1}
			\begin{pmatrix}
				0 & 0 \\
				1 & 0
			\end{pmatrix}
			\right] \uparrow  \notag\\
			&\frac{	\theta_0^{2}}{\prod\limits_{l=1}^{M_2} 2^{\tilde{m}_{l}^{1}}}
			\uparrow^* \rprod_{l=1}^{r_2}
			\left[
			\begin{pmatrix}
				y^{-\frac{1}{2}} & 0 \\
				0 & y^{\frac{1}{2}}
			\end{pmatrix}^{q_l^2 - p_{l-1}^2 - 1}
			\begin{pmatrix}
				0 & (1 + (-1)^{\tilde{m}_l^{2}-1}) \\
				0 & 0
			\end{pmatrix}
			\begin{pmatrix}
				y^{-\frac{1}{2}} & 0 \\
				0 & y^{\frac{1}{2}}
			\end{pmatrix}^{p_l^2 - q_{l-1}^2 - 1}
			\begin{pmatrix}
				0 & 0 \\
				1 & 0
			\end{pmatrix}
			\right] \uparrow  \notag\\
			=&\theta_0^{1}\theta_0^{2}\prod_{j=1}^{2}\prod_{k=1}^{M_j} 2^{-\tilde{m}_k^{j}}
			x^{-\frac{M_1}{2}}
			\prod_{i=1}^{M_1}
			2^{\delta_{(\tilde{m}_i^1 - \tilde{n}_i^1), 1}}
			x^{i(\tilde{m}_i^1 - \tilde{n}_i^1)}
			\cdot
			y^{-\frac{M_2}{2}} \prod_{l=1}^{M_2}
			2^{\delta_{(\tilde{m}_l^2 - \tilde{n}_l^2), 1}}
			y^{l(\tilde{m}_l^2 - \tilde{n}_l^2)},
		\end{align}
		where $q_0^{j}=-1$, $q_{r_j}^{j}=M_j +1$ and
		$\rprod\limits_{l=1}^{r_2}$  means that the subscripts decrease in the direction of the arrow.
		Therefore,
		\begin{align}
			&\prescript{(2)}{}{\langle \tilde{m}^{2}|}\bigotimes \prescript{(1)}{}{\langle \tilde{m}^{1}|}
			\tilde{\mathbb{B}}_{1}(x)\tilde{\mathbb{B}}_{2}(y)
			|\tilde{n}^{1}\rangle^{(1)}\bigotimes|\tilde{n}^{2}\rangle^{(2)}\notag\\
			=&
			\theta_0^{1}\theta_0^{2}\prod_{j=1}^{2}\prod_{k=1}^{M_j} 2^{-\tilde{m}_k^{j}}
			\prod_{i=1}^{M_1}
			2^{\delta_{(\tilde{m}_i^1 - \tilde{n}_i^1), 1}}
			x^{i(\tilde{m}_i^1 - \tilde{n}_i^1)}
			\cdot
			\prod_{l=1}^{M_2}
			2^{\delta_{(\tilde{m}_l^2 - \tilde{n}_l^2), 1}}
			y^{l(\tilde{m}_l^2 - \tilde{n}_l^2)}.
		\end{align}
		In view of  Eq.$(\ref{tgibm-scalar-product})$ and Eq.(\ref{xrxtj}), Eq.$(\ref{1-B-n,C-n,2})$ can be derived.
		Analogous to the proof of Eq.$(\ref{1-B-n,C-n,2})$,   we can prove Eq.$(\ref{2-B-n,C-n,2})$.
	\end{proof}
	Similarly,  the following equations  hold
	\begin{align}
		\tilde{\mathbb{B}}_{1}(x)|\tilde{n}^{1}\rangle^{(1)}&=\sum_{|\tilde{m}^{1} \rangle^{(1)} \triangleright | \tilde{n}^{1} \rangle^{(1)} }\prod_{i=1}^{M_1}
		2^{\delta_{(\tilde{m}_i^1 - \tilde{n}_i^1), 1}}
		x^{i(\tilde{m}^{1}_i - \tilde{n}^{1}_i)} | \tilde{m}^{1} \rangle ^{(1)}, \label{B1-1}
		\\
		\tilde{\mathbb{B}}_{2}(y)|\tilde{n}^{2}\rangle^{(2)}&=\sum_{|\tilde{m}^{2} \rangle^{(2)} \triangleright | \tilde{n}^{2} \rangle^{(2)} }
		\prod_{j=1}^{M_2} 2^{\delta_{(\tilde{m}_j^2 - \tilde{n}_j^2), 1}}
		y^{j(\tilde{m}^{2}_j - \tilde{n}^{2}_j)} | \tilde{m}^{2} \rangle^{(2)},\label{B2-1}\\
		\prescript{(1)}{}{\langle \tilde{n}^{1}|}\tilde{\mathbb{C}}_{1}(x)
		&=\sum_{\prescript{(1)}{}{\langle \tilde{n}^{1}|}
			\triangleleft
			\prescript{(1)}{}{\langle \tilde{m}^{1}|} }
		\prod_{i=1}^{M_1} 2^{\delta_{(\tilde{m}_i^1 - \tilde{n}_i^1), 1}}x^{i(\tilde{m}^{1}_i - \tilde{n}^{1}_i)}\prescript{(1)}{}{\langle \tilde{m}^{1}|},\label{C1-1}\\
		\prescript{(2)}{}{\langle \tilde{n}^{2}|}\tilde{\mathbb{C}}_{2}(y)
		&=\sum_{\prescript{(2)}{}{\langle \tilde{n}^{2}|}
			\triangleleft
			\prescript{(2)}{}{\langle \tilde{m}^{2}|} }
		\prod_{j=1}^{M_2} 2^{\delta_{(\tilde{m}_j^2 - \tilde{n}_j^2), 1}}
		y^{j(m^{2}_j - \tilde{n}^{2}_j)}  \prescript{(2)}{}{\langle \tilde{m}^{2}|} \label{C2-1}.
	\end{align}
	
	For an arbitrary pair of strict partitions $\tilde{\mu}$ and $\tilde{\nu}$,
	the skew Schur $Q$-function is defined as
	\begin{align}\label{skew-schur}
		Q_{\tilde{\mu}/\tilde{\nu}}(x) =
		\begin{cases}
			2^{\#(\tilde{\mu}|\tilde{\nu})} x^{|\tilde{\mu}| - |\tilde{\nu}|}, & \tilde{\mu} \succ \tilde{\nu}, \\
			0, & \text{otherwise}.
		\end{cases}
	\end{align}
	In the case of $\tilde{\nu}=\emptyset$, the skew Schur $Q$-function $Q_{\tilde{\mu}/\tilde{\nu}}(x) $ is reduced to Schur $Q$-function $Q_{\tilde{\mu}}(x) $ that is given by
	\begin{align}
		Q_{\tilde{\mu}}\{\mathbf{x}\} = \operatorname{Pf}\left( q_{\tilde{\mu}_i}\{\mathbf{x}\} q_{\tilde{\mu}_j}\{\mathbf{x}\} + 2 \sum_{k=1}^{\tilde{\mu}_j} (-)^k q_{(\tilde{\mu}_i + k)}\{\mathbf{x}\} q_{(\tilde{\mu}_j - k)}\{\mathbf{x}\} \right)_{\substack{1 \leq i < j \leq 2r}},
	\end{align}
	where $\{\mathbf{x}\}=\{x_1,\cdots,x_{n}\}$,
	\begin{align}
		q_{m}\{\mathbf{x}\} = \operatorname{Coeff}_{k^m}\left[ \prod_{i=1}^{\infty} \frac{1 + x_i k}{1 - x_i k} \right]
	\end{align}
	and $Pf(A)$ denotes a Pfaffian of skew symmetric matrix $A=(a_{i,j})_{2r \times 2r}$ and is  of the form
	\begin{align}
		Pf(A) = \sum_{\omega \in \mathcal{S}_{2n}} sgn(\omega) a_{\omega(1)\omega(2)} \cdots a_{\omega(2n-1)\omega(2n)},
	\end{align}
	which summed over $2r$-order permutation group $\omega \in S_{2r}$ such that $\omega(2i-1) < \omega(2i)$ for $1 \leq i \leq r$, and $\omega(2i-1) < \omega(2i+1)$ for $1 \leq i \leq r-1.$
	The following  identity holds
	\begin{align}\label{schur-skew-schurQ}
		Q_{\tilde{\mu}}\{x_1, \ldots, x_n\} = \sum_{\tilde{\nu} \subset [n-1, \infty]} Q_{\tilde{\mu}/\tilde{\nu}}(x_n) Q_{\tilde{\nu}}\{x_1, \ldots, x_{n-1}\}.	
	\end{align}

	\begin{proposition}
		Set 
		\begin{align}\label{M-map}
			\tilde{\mathcal{M}}  | \tilde{n}^{1}\rangle^{(1)}\bigotimes|\tilde{n}^{2}\rangle^{(2)}&=
			2^{-l(\tilde{\nu}^{1})-l(\tilde{\nu}^{2})}	| \tilde{\nu}^{1},\tilde{\nu}^{2} \rangle ,
			\quad 	\tilde{\mathcal{M}}  | m^{1}\rangle^{(1)}\bigotimes|m^{2}\rangle^{(2)}=
			2^{-l(\tilde{\nu}^{1})-l(\tilde{\nu}^{2})}
			| \tilde{\mu}^{1},\tilde{\mu}^{2} \rangle,
			\\
			\prescript{(2)}{}{\langle \tilde{n}^{2}|}\bigotimes \prescript{(1)}{}{\langle \tilde{n}^{1}|}  \tilde{\mathcal{M}}^* &= 2^{-l(\tilde{\nu}^{1})-l(\tilde{\nu}^{2})}
			\langle \tilde{\nu}^{2},\tilde{\nu}^{1}|, \quad
			\prescript{(2)}{}{\langle \tilde{m}^{2}|}\bigotimes \prescript{(1)}{}{\langle \tilde{m}^{1}|}  \tilde{\mathcal{M}}^* = 2^{-l(\tilde{\nu}^{1})-l(\tilde{\nu}^{2})}
			\langle \tilde{\mu}^{2},\tilde{\mu}^{1}| ,
		\end{align}
		we have
		\begin{align}
			&\tilde{\mathcal{M}} \tilde{\mathbb{B}}_{1}(x)\tilde{\mathbb{B}}_{2}(y)|\tilde{n}^{1}\rangle^{(1)}\bigotimes|\tilde{n}^{2}\rangle^{(2)} \notag\\
			=& \sum_{{ \tilde{\nu}^{1} \prec \tilde{\mu}^{1} \subseteq [l_{1}+1,M_{1}] }\atop{\tilde{\nu}^{2} \prec \tilde{\mu}^{2} \subseteq [l_{2}+1,M_{2}]}}
			2^{\#(\tilde{\mu}^{1}|\tilde{\nu}^{1}) - l(\tilde{\mu}^{1})}
			2^{\#(\tilde{\mu}^{2}\tilde{\nu}^{2}) - l(\tilde{\mu}^{2})}
			x^{|\tilde{\mu}^{1}| - |\tilde{\nu}^{1}|} y^{|\tilde{\mu}^{2}| - |\tilde{\nu}^{2}|}
			| \tilde{\mu}^{1},\tilde{\mu}^{2} \rangle	\notag\\
			= &\sum_{{ \tilde{\nu}^{1} \prec \tilde{\mu}^{1} \subseteq [l_{1}+1,M_{1}] }\atop{\tilde{\nu}^{2} \prec \tilde{\mu}^{2} \subseteq [l_{2}+1,M_{2}]}}
			2^{- l(\tilde{\mu}^{1})}	Q_{\tilde{\mu}^{1}/\tilde{\nu}^{1}}(x) 2^{- l(\tilde{\mu}^{2})}Q_{\tilde{\mu}^{2}/\tilde{\nu}^{2}}(y)
			| \tilde{\mu}^{1},\tilde{\mu}^{2} \rangle,	\label{M-B1-B2-N}
			\\
			&\prescript{(2)}{}{\langle \tilde{n}^{2}|}\bigotimes \prescript{(1)}{}{\langle \tilde{n}^{1}|}
			\tilde{\mathbb{C}}_{2}(y)\tilde{\mathbb{C}}_{1}(x)
			\tilde{\mathcal{M}}^* \notag\\
			=&
			\sum_{{ \tilde{\nu}^{1} \prec \tilde{\mu}^{1} \subseteq [l_{1}+1,M_{1}] }\atop{\tilde{\nu}^{2} \prec \tilde{\mu}^{2} \subseteq [l_{2}+1,M_{2}]}}
			2^{\#(\tilde{\mu}^{1}|\tilde{\nu}^{1}) - l(\tilde{\mu}^{1})}
			2^{\#(\tilde{\mu}^{2}\tilde{\nu}^{2}) - l(\tilde{\mu}^{2})}
			x^{|\tilde{\mu}^{1}| - |\tilde{\nu}^{1}|} y^{|\tilde{\mu}^{2}| - |\tilde{\nu}^{2}|}
			\langle \tilde{\mu}^{2},\tilde{\mu}^{1} |\notag\\
			=&	\sum_{{ \tilde{\nu}^{1} \prec \tilde{\mu}^{1} \subseteq [l_{1}+1,M_{1}] }\atop{\tilde{\nu}^{2} \prec \tilde{\mu}^{2} \subseteq [l_{2}+1,M_{2}]}}
			2^{- l(\tilde{\mu}^{1})}	Q_{\tilde{\mu}^{1}/\tilde{\nu}^{1}}(x) 2^{- l(\tilde{\mu}^{2})}Q_{\tilde{\mu}^{2}/\tilde{\nu}^{2}}(y)
			\langle \tilde{\mu}^{2},\tilde{\mu}^{1} |,\label{M-C1-C2-N}
		\end{align}	
		which are summed over interlacing  strict boxed   $2$-partitions.
	\end{proposition}
	\begin{proof}
		Since the occupation number takes a value of $0$ or $1$ in positions $1$ to
		$M_j$,  $\tilde{m}_i^{j}= \tilde{n}_i^{j} + 1$
		implies that $\tilde{m}_i^{j}= 1$ and $n_i^{j} = 0$.
		It can be seen from the Definition $\ref{M}$ that
		the element $i^{\tilde{m}_i^{j}}$  exists in strict partition $\tilde{\mu}^{j}$ but not in $\tilde{\nu}^{j}$ when  $\tilde{m}_i^{j}= \tilde{n}_i^{j} + 1$  is satisfied.
		Combining the Lemma $\ref{B-n,C-n,2}$ and Eq.$(\ref{skew-schur})$
		leads to Eq.$(\ref{M-B1-B2-N})$.
		The proof of Eq.$(\ref{M-C1-C2-N})$  is similar and will not be repeated here.
	\end{proof}

	\section{The scalar product and boxed BUC plane partitions}
	The aim of this section is
	to investigate the  generating function for boxed BUC plane partitions by means of the scalar product for the   generalized $i$-boson model.
	It is found that the generating function can be rewritten as  products of   Schur $Q$-functions.

	\begin{definition}
		The boxed BUC plane partition refers to the BUC plane partition placed into a box of a finite size.
		Set the   strict boxed   $2$-partition $(\tilde{\chi}_{i})=(\tilde{\alpha}^i,\tilde{\beta}^i)$, $i\in\mathbf{Z}$.	 Note the boxed BUC plane partition as $\tilde{\Pi}^{'}=(\ldots,\tilde{\chi}_{-1},\tilde{\chi}_{0},\tilde{\chi}_{1},\ldots)$, which  satisfies Eq.$(\ref{BUC-interlacing-partitions})$,
		and  denotes a pair of strict boxed  plane partitions
		$\tilde{\pi}^{1}=(\ldots,\tilde{\alpha}^{-1},\tilde{\alpha}^{0},\tilde{\alpha}^{1},\ldots)$ and $\tilde{\pi}^{2}=(\ldots,\tilde{\beta}^{-1},\tilde{\beta}^{0},\tilde{\beta}^{1},\ldots)$.
	\end{definition}
	
	\begin{lemma}\label{M-n-0}
		The following equations hold
		\begin{align}
			&\tilde{\mathcal{M}}
			\lprod_{i=1}^{N_1} \tilde{\mathbb{B}}_{1}(x_i)
			\lprod_{j=1}^{N_2} \tilde{\mathbb{B}}_{2}(y_j)
			|0\rangle^{(1)}\bigotimes|0\rangle^{(2)}\notag\\
			=& \sum_{{ [N_{1},M_{1}] \ni \tilde{\pi}^{1}_0 \succ \cdots \succ \tilde{\pi}^{1}_{N_{1}} = \emptyset}
				\atop{[N_{2},M_{2}] \ni \tilde{\pi}^{2}_0 \succ \cdots \succ \tilde{\pi}^{2}_{N_{2}} = \emptyset}}
			2^{-l(\tilde{\pi}^{1}_{0})}
			\prod_{i=1}^{N_1}
			2^{\#(\tilde{\pi}^{1}_{i-1}|\tilde{\pi}^{1}_{i}) }		
			x_i^{|\tilde{\pi}^{1}_{i-1}| - |\tilde{\pi}^{1}_i|}
			2^{-l(\tilde{\pi}^{2}_{0})}	
			\prod_{j=1}^{N_2}
			2^{\#(\tilde{\pi}^{2}_{j-1}|\tilde{\pi}^{2}_{j}) }		
			y_j^{|\tilde{\pi}^{2}_{j-1}| - |\tilde{\pi}^{2}_j|}
			| \tilde{\pi}^{1}_0,\tilde{\pi}^{2}_0 \rangle,
			\label{M-B1-BN-x}\\
			&\prescript{(2)}{}{\langle 0|}\bigotimes \prescript{(1)}{}{\langle 0|}
			\rprod_{j=1}^{N_2} \tilde{\mathbb{C}}_{2}(y_j)
			\rprod_{i=1}^{N_1} \tilde{\mathbb{C}}_{1}(x_i) \tilde{\mathcal{M}}^* \notag\\
			=&
			\sum_{{ \emptyset=\tilde{\pi}^{1}_{-N_{1}} \prec \cdots \prec \tilde{\pi}^{1}_0  \in [N_{1}, M_{1}]}
				\atop{\emptyset=\tilde{\pi}^{2}_{-N_{2}} \prec \cdots \prec \tilde{\pi}^{2}_0  \in [N_{2}, M_{2}]}} 		
			2^{-l(\tilde{\pi}^{1}_{0})}
			\prod_{i=1}^{N_1}
			2^{\#(\tilde{\pi}^{1}_{-i+1}|\tilde{\pi}^{1}_{-i}) }		
			x_i^{|\tilde{\pi}^{1}_{-i+1}| - |\tilde{\pi}^{1}_{-i}|}
			2^{-l(\tilde{\pi}^{2}_{0})}	
			\prod_{j=1}^{N_2}
			2^{\#(\tilde{\pi}^{2}_{-j+1}|\tilde{\pi}^{2}_{-j}) }		
			y_j^{|\tilde{\pi}^{2}_{-j+1}| - |\tilde{\pi}^{2}_{-j}|}
			\langle \tilde{\pi}^{2}_0,\tilde{\pi}^{1}_0|,\label{M-C1-CN-x}
		\end{align}	
		which are summed over $[N_{j},M_{j}] \ni \tilde{\pi}^{j}_0 \succ \cdots \succ \tilde{\pi}^{j}_{N_{j}} = \emptyset$ and $\emptyset=\tilde{\pi}^{j}_{-N_{j}} \prec \cdots \prec\tilde{\pi}^{j}_0  \in [N_{j}, M_{j}]$ respectively.
	\end{lemma}
	\begin{proof}
		Setting  $|\tilde{n}^{1}\rangle^{(1)}\bigotimes|\tilde{n}^{2}\rangle^{(2)}=|0\rangle^{(1)}\bigotimes|0\rangle^{(2)}$, by continuously applying Eq.$(\ref{B1-1})$ $N_1$ times and Eq.$(\ref{B2-1})$ $N_2$ times respectively, one can get
		\begin{align}
			&\lprod_{i=1}^{N_1} \tilde{\mathbb{B}}_{1}(x_i)
			\lprod_{j=1}^{N_2} \tilde{\mathbb{B}}_{2}(y_j)
			|0\rangle^{(1)}\bigotimes|0\rangle^{(2)}\notag\\
			=&\sum_{{|m^{a}\rangle^{(1)} \triangleright |m^{a+1}\rangle^{(1)}}
				\atop{a=1,\cdots,N_1}}
			\sum_{{|n^{b}\rangle^{(2)} \triangleright |n^{b+1}\rangle^{(2)}}
				\atop{b=1,\cdots,N_2}}
			\prod_{l=1}^{N_1}\prod_{i=1}^{M_1}
			2^{\delta_{(\tilde{m}^{l}_i - \tilde{m}^{l+1}_i), 1}}
			x^{i(\tilde{m}^{l}_i - \tilde{m}^{l+1}_i)}\notag\\
			&	\prod_{k=1}^{N_2}\prod_{j=1}^{M_2}
			2^{\delta_{(\tilde{n}^{k}_j - \tilde{n}^{k+1}_j), 1}}
			y^{j(\tilde{n}^{k}_j - \tilde{n}^{k+1}_j)}
			|m^{1}\rangle^{(1)}\bigotimes|\tilde{n}^{1}\rangle^{(2)},
		\end{align}
		where $|\tilde{m}^{N_1+1}\rangle^{(1)}=|0\rangle^{(1)}$ and  $|\tilde{n}^{N_2+1}\rangle^{(2)}=|0\rangle^{(2)}$.
		By using the assumption
		\begin{align}
			\tilde{\mathcal{M}} |\tilde{m}^{a}\rangle^{(1)}\bigotimes|\tilde{n}^{b}\rangle^{(2)}
			=|\tilde{\pi}_{a-1}^{1},\tilde{\pi}_{b-1}^{2}\rangle,
		\end{align}
		it is simple to obtain that
		\begin{align}
			&\tilde{\mathcal{M}}
			\lprod_{i=1}^{N_1} \tilde{\mathbb{B}}_{1}(x_i)
			\lprod_{j=1}^{N_2} \tilde{\mathbb{B}}_{2}(y_j)
			|0\rangle^{(1)}\bigotimes|0\rangle^{(2)}\notag\\
			= &\sum_{{ [N_{1},M_{1}] \ni \tilde{\pi}^{1}_0 \succ \cdots \succ \tilde{\pi}^{1}_{N_{1}} = \emptyset}
				\atop{[N_{2},M_{2}] \ni \tilde{\pi}^{2}_0 \succ \cdots \succ \tilde{\pi}^{2}_{N_{2}} = \emptyset}} 	
			2^{-l(\tilde{\pi}^{1}_{0})}
			\prod_{i=1}^{N_1}
			2^{\#(\tilde{\pi}^{1}_{i-1}|\tilde{\pi}^{1}_{i}) }		
			x_i^{|\tilde{\pi}^{1}_{i-1}| - |\tilde{\pi}^{1}_i|}
			2^{-l(\tilde{\pi}^{2}_{0})}	
			\prod_{j=1}^{N_2}
			2^{\#(\tilde{\pi}^{2}_{j-1}|\tilde{\pi}^{2}_{j}) }		
			y_j^{|\tilde{\pi}^{2}_{j-1}| - |\tilde{\pi}^{2}_j|}
			| \tilde{\pi}^{1}_0,\tilde{\pi}^{2}_0 \rangle.
		\end{align}	
		According to the similar approach, Eq.$(\ref{M-C1-C2-N})$ can also be proved.
	\end{proof}

	\begin{lemma}\label{M-n-0-schur}
		Let
		$\{\mathbf{a^j}\}=\{a_1,\cdots,a_{N_j}\}$  and $\{\bar{\mathbf{a}}^j\}=\{a_1,\cdots,a_{N_j-1}\}$, where $\mathbf{a}$ can be assigned as $\mathbf{x},\mathbf{y},\mathbf{z},\mathbf{v}$.
		The following equations hold
		\begin{align}
			\tilde{\mathcal{M}}
			\lprod_{i=1}^{N_1} \tilde{\mathbb{B}}_{1}(x_i)
			\lprod_{j=1}^{N_2} \tilde{\mathbb{B}}_{2}(y_j)
			|0\rangle^{(1)}\bigotimes|0\rangle^{(2)}
			= &\sum_{{ \tilde{\mu}^{1} \subseteq [N_1,M_1]}\atop{\tilde{\mu}^{2} \subseteq [N_2,M_2]}}
			2^{-l(\tilde{\mu}^{1})}	2^{-l(\tilde{\mu}^{2})}
			Q_{\tilde{\mu}^{1}}\{\mathbf{x^1}\}  Q_{\tilde{\mu}^{2}}\{\mathbf{y^2}\}
			| \tilde{\mu}^{1},\tilde{\mu}^{2} \rangle,\label{M-B1-BN-schur}	\\	
			\prescript{(2)}{}{\langle 0|}\bigotimes \prescript{(1)}{}{\langle 0|}
			\rprod_{j=1}^{N_2} \tilde{\mathbb{C}}_{2}(y_j)
			\rprod_{i=1}^{N_1} \tilde{\mathbb{C}}_{1}(x_i)\tilde{\mathcal{M}}^*
			= &\sum_{{ \tilde{\mu}^{1} \subseteq [N_1,M_1]}\atop{\tilde{\mu}^{2} \subseteq [N_2,M_2]}}
			2^{-l(\tilde{\mu}^{1})}	2^{-l(\tilde{\mu}^{2})}
			Q_{\tilde{\mu}^{1}}\{\mathbf{x^1}\}  Q_{\tilde{\mu}^{2}}\{\mathbf{y^2}\}
			\langle \tilde{\mu}^{2},\tilde{\mu}^{1}|.\label{M-C1-CN-schur}
		\end{align}	
	\end{lemma}
	\begin{proof}
		
		We prove Eq.$(\ref{M-B1-BN-schur})$ by the   induction, Eq.$(\ref{M-C1-CN-schur})$  can be proved similarly.	
		Based on the special case of $|\tilde{n}^{1}\rangle^{(1)}\bigotimes|\tilde{n}^{2}\rangle^{(2)}=|0\rangle^{(1)}\bigotimes|0\rangle^{(2)}$ in Eq.$(\ref{M-C1-C2-N})$, one yields
		\begin{align}
			\tilde{\mathcal{M}} \tilde{\mathbb{B}}_{1}(x)\tilde{\mathbb{B}}_{2}(y) |0\rangle^{(1)}\bigotimes|0\rangle^{(2)}
			=& \sum_{{ \emptyset \prec \tilde{\mu}^{1} \subseteq [1,M_{1}]}\atop{\emptyset \prec \tilde{\mu}^{2} \subseteq [1,M_{2}]}}
			2^{- l(\tilde{\mu}^{1})}	2^{- l(\tilde{\mu}^{2})}	
			Q_{{\tilde{\mu}^{1}}/\emptyset}(x) 	Q_{{\tilde{\mu}^{2}}/\emptyset}(y)
			|\tilde{\mu}^{1},\tilde{\mu}^{2}\rangle \notag\\
			=&\sum_{{ \emptyset \prec \tilde{\mu}^{1} \subseteq [1,M_{1}]}\atop{\emptyset \prec \tilde{\mu}^{2} \subseteq [1,M_{2}]}}
			2^{- l(\tilde{\mu}^{1})}		2^{-l(\tilde{\mu}^{2})}	
			Q_{\tilde{\mu}^{1}}(x)  	Q_{\tilde{\mu}^{2}}(y)  |\tilde{\mu}^{1},\tilde{\mu}^{2}\rangle .
		\end{align}
		For all $N_{1}, N_{2}\geq 2$, assume that
		\begin{align}\label{assumption-N-1}
			\tilde{\mathcal{M}}
			\prod_{i=1}^{N_{1}-1} \tilde{\mathbb{B}}_{1}(x_i)
			\prod_{j=1}^{N_{2}-1} \tilde{\mathbb{B}}_{2}(y_i)
			|0\rangle^{(1)}\bigotimes|0\rangle^{(2)}
			= &\sum_{{ \tilde{\nu}^{1} \subseteq [N_{1}-1,M_1]}\atop{\tilde{\nu}^{2} \subseteq [N_{2}-1,M_2]}}
			2^{- l(\tilde{\nu}^{1})}	2^{- l(\tilde{\nu}^{2})}			
			Q_{\tilde{\nu}^{1}}\{\bar{\mathbf{x}}^1\}  Q_{\tilde{\nu}^{2}}\{ \bar{\mathbf{y}}^2 \}
			| \tilde{\nu}^{1},\tilde{\nu}^{2} \rangle.
		\end{align}	
		By using Eqs.$(\ref{M})$ and $(\ref{assumption-N-1})$, it is easy to show that
		\begin{align}
			\prod_{i=1}^{N_{1}-1} \tilde{\mathbb{B}}_{1}(x_i)
			\prod_{j=1}^{N_{2}-1} \tilde{\mathbb{B}}_{2}(y_i)
			|0\rangle^{(1)}\bigotimes|0\rangle^{(2)}
			= \sum_{{ 	| \tilde{n}^{1}\rangle^{(1)}|\Sigma_0^{\tilde{n}^{1}} =N_{1}-1}\atop{| \tilde{n}^{2}\rangle^{(2)}|\Sigma_0^{\tilde{n}^{2}}=N_{2}-1} }
			Q_{\tilde{\nu}^{1}}\{\bar{\mathbf{x}}^1\}  Q_{\tilde{\nu}^{2}}\{ \bar{\mathbf{y}}^2 \}
			| \tilde{n}^{1}\rangle^{(1)}\bigotimes|\tilde{n}^{2}\rangle^{(2)},
		\end{align}	
		where  $\Sigma_0^{n^j} = \sum\limits_{k=0}^{M_j} n_k^j=N_{j}-1$ $(j=1,2)$.
		According to   Eq.$(\ref{B,C,commute})$  and
		applying $\tilde{\mathcal{M}} \tilde{\mathbb{B}}_{1}(x_{N_1})\tilde{\mathbb{B}}_{2}(y_{N_2})$ to the above equation, one gets
		\begin{align}
			&\tilde{\mathcal{M}} \tilde{\mathbb{B}}_{1}(x_{N_1})\tilde{\mathbb{B}}_{2}(y_{N_2})
			\prod_{i=1}^{N_{1}-1} \tilde{\mathbb{B}}_{1}(x_i)
			\prod_{j=1}^{N_{2}-1} \tilde{\mathbb{B}}_{2}(y_i)	
			|0\rangle^{(1)}\bigotimes|0\rangle^{(2)}\notag\\
			=&\sum_{{ \tilde{\nu}^{1} \subseteq [N_{1}-1,M_1]}\atop{\tilde{\nu}^{2} \subseteq [N_{2}-1,M_2]}}
			Q_{\tilde{\nu}^{1}}\{\bar{\mathbf{x}}^1\}  Q_{\tilde{\nu}^{2}}\{ \bar{\mathbf{y}}^2 \} 	
			\sum_{{\tilde{\nu}^{1} \prec \tilde{\mu}^{1} \subseteq [N_1,M_1]}\atop{\tilde{\nu}^{2} \prec \tilde{\mu}^{2} \subseteq [N_2,M_2]}}
			2^{- l(\tilde{\mu}^{1})}	2^{- l(\tilde{\mu}^{2})}	
			Q_{\tilde{\mu}^{1}/\tilde{\nu}^{1}}(x_{N_1}) Q_{\tilde{\mu}^{2}/\tilde{\nu}^{2}}(y_{N_2}) 	
			| \tilde{\mu}^{1},\tilde{\mu}^{2} \rangle.
		\end{align}	
		It follows from the interlacing relation $\tilde{\nu}^{j} \prec \tilde{\mu}^{j}$ and Eq.$(\ref{schur-skew-schurQ})$ that
		\begin{align}
			&\tilde{\mathcal{M}}
			\lprod_{i=1}^{N_1} \tilde{\mathbb{B}}_{1}(x_i)
			\lprod_{j=1}^{N_2} \tilde{\mathbb{B}}_{2}(y_j)
			|0\rangle^{(1)}\bigotimes|0\rangle^{(2)}\notag\\
			= &	\sum_{{ \tilde{\mu}^{1} \subseteq [N_1,M_1]}\atop{\tilde{\mu}^{2} \subseteq [N_2,M_2]}}
			\sum_{{ \tilde{\nu}^{1} \subseteq [N_{1}-1,\infty]}\atop{\tilde{\nu}^{2} \subseteq [N_{2}-1,\infty]}}
			2^{- l(\tilde{\mu}^{1})}	2^{- l(\tilde{\mu}^{2})}	
			Q_{\tilde{\mu}^{1}/\tilde{\nu}^{1}}(x_{N_1}) 	Q_{\tilde{\nu}^{1}}\{\bar{\mathbf{x}}^1\}
			Q_{\tilde{\mu}^{2}/\tilde{\nu}^{2}}(y_{N_2}) 	Q_{\tilde{\nu}^{2}}\{ \bar{\mathbf{y}}^2 \}
			| \tilde{\mu}^{1},\tilde{\mu}^{2} \rangle \notag\\	
			= &\sum_{{ \tilde{\mu}^{1} \subseteq [N_1,M_1]}\atop{\tilde{\mu}^{2} \subseteq [N_2,M_2]}}
			2^{- l(\tilde{\mu}^{1})}	2^{- l(\tilde{\mu}^{2})}	
			Q_{\tilde{\mu}^{1}}\{\mathbf{x^1}\}  Q_{\tilde{\mu}^{2}}\{\mathbf{y^2}\}
			| \tilde{\mu}^{1},\tilde{\mu}^{2} \rangle.	
		\end{align}	
		Thus, Eq.$(\ref{M-B1-BN-schur})$ holds for any $N_{1},N_{2} \geq 1$.
	\end{proof}

Consider the scalar product of the   generalized $i$-boson model
\begin{align}
	&\tilde{S}(\{\mathbf{x^1}\}, \{\mathbf{y^2}\}, \{\mathbf{z^1}\}, \{\mathbf{v^2}\})\notag\\
	=&\bigg \langle \prescript{(2)}{}{\langle 0|}\bigotimes \prescript{(1)}{}{\langle 0|}
	\rprod_{j=1}^{N_2} \tilde{\mathbb{C}}_{2}(y_j)
	\prod_{i=1}^{N_1} \tilde{\mathbb{C}}_{1}(x_i) \tilde{\mathcal{M}}^*
	,\tilde{\mathcal{M}}\prod_{l=1}^{N_1} \tilde{\mathbb{B}}_{1}(z_i)
	\prod_{k=1}^{N_2} \tilde{\mathbb{B}}_{2}(v_k)
	|0\rangle^{(1)}\bigotimes|0\rangle^{(2)} \bigg\rangle.
\end{align}

	\begin{proposition}
		The generating function of boxed BUC plane partitions can expressed as
		\begin{align}\label{schs}
			\tilde{S}(\{\mathbf{x^1}\}, \{\mathbf{y^2}\}, \{\mathbf{z^1}\}, \{\mathbf{v^2}\})
			=&\sum_{{\tilde{\pi}^{1} \subseteq [N_{1}, N_{1}, M_{1}]}
				\atop{\tilde{\pi}^{2} \subseteq [N_{2}, N_{2}, M_{2}]}}
			B_{\tilde{\pi}^{1} }(\{\mathbf{x^1}\}, \{\mathbf{z^1}\}) B_{\tilde{\pi}^{2} }(\{\mathbf{y^2}\}, \{\mathbf{v^2}\})\notag\\
			=&\sum_{{ \tilde{\mu}^{1} \subseteq [N_1,M_1]}\atop{\tilde{\mu}^{2} \subseteq [N_2,M_2]}}
			2^{- l(\tilde{\mu}^{1})}2^{- l(\tilde{\mu}^{2})}
			Q_{\tilde{\mu}^{1}}\{\mathbf{x^1}\}  Q_{\tilde{\mu}^{1}}\{\mathbf{z^1}\}
			Q_{\tilde{\mu}^{2}}\{\mathbf{y^2}\}  Q_{\tilde{\mu}^{2}}\{\mathbf{v^2}\},
		\end{align}
		where
		\begin{align}
			B_{\tilde{\pi}^{1} }(\{\mathbf{x^1}\}, \{\mathbf{z^1}\})
			&=2^{p(\tilde{\pi}^1)}
			\prod_{i=1}^{N_1}
			x_i^{|\tilde{\pi}^{1}_{-i+1}| - |\tilde{\pi}^{1}_{-i}|}
			z_i^{|\tilde{\pi}^{1}_{i-1}| - |\tilde{\pi}^{1}_i|},   \\
			B_{\tilde{\pi}^{2} }(\{\mathbf{y^2}\}, \{\mathbf{v^2}\})
			&=2^{p(\tilde{\pi}^2)}
			\prod_{j=1}^{N_2}
			y_j^{|\tilde{\pi}^{2}_{-j+1}| - |\tilde{\pi}^{2}_{-j}|}
			v_j^{|\tilde{\pi}^{2}_{j-1}| - |\tilde{\pi}^{2}_j|}.
		\end{align}
	\end{proposition}
	
	\begin{proof}
		Combining the Lemma $\ref{strict-pi}$, Eq.$(\ref{tgibm-scalar-product})$ and the Lemma $\ref{M-n-0}$ yields
		\begin{align}\label{scalar-product-A1-A2}
			&S(\{\mathbf{x^1}\}, \{\mathbf{y^2}\}, \{\mathbf{z^1}\}, \{\mathbf{v^2}\})\notag\\
			=&
			\sum_{{\tilde{\pi}^{1} \subseteq [N_{1}, N_{1}, M_{1}]}
				\atop{\tilde{\pi}^{2} \subseteq [N_{2}, N_{2}, M_{2}]}}
			\prod_{i=1}^{N_1}
			2^{\#(\tilde{\pi}^{1}_{-i+1}|\tilde{\pi}^{1}_{-i}) }		
			x_i^{|\tilde{\pi}^{1}_{-i+1}| - |\tilde{\pi}^{1}_{-i}|}
			\prod_{j=1}^{N_2}
			2^{\#(\tilde{\pi}^{2}_{-j+1}|\tilde{\pi}^{2}_{-j}) }		
			y_j^{|\tilde{\pi}^{2}_{-j+1}| - |\tilde{\pi}^{2}_{-j}|}
			\notag\\
			&
			\prod_{l=1}^{N_1}
			2^{\#(\tilde{\pi}^{1}_{l-1}|\tilde{\pi}^{1}_{l}) }		
			z_l^{|\tilde{\pi}^{1}_{l-1}| - |\tilde{\pi}^{1}_l|}
			\prod_{k=1}^{N_2}
			2^{\#(\tilde{\pi}^{2}_{k-1}|\tilde{\pi}^{2}_{k}) }		
			v_k^{|\tilde{\pi}^{2}_{k-1}| - |\tilde{\pi}^{2}_k|}
			\cdot 2^{-l(\tilde{\pi}^{1}_{0})} 2^{-l(\tilde{\pi}^{2}_{0})}
			\notag\\
			=&2^{p(\tilde{\pi}^1)}2^{p(\tilde{\pi}^2)}
			\sum_{{\tilde{\pi}^{1} \subseteq [N_{1}, N_{1}, M_{1}]}
				\atop{\tilde{\pi}^{2} \subseteq [N_{2}, N_{2}, M_{2}]}}
			\prod_{i=1}^{N_1}
			x_i^{|\tilde{\pi}^{1}_{-i+1}| - |\tilde{\pi}^{1}_{-i}|}
			z_i^{|\tilde{\pi}^{2}_{i-1}| - |\tilde{\pi}^{2}_i|}
			\prod_{j=1}^{N_2}
			y_j^{|\tilde{\pi}^{1}_{-j+1}| - |\tilde{\pi}^{1}_{-j}|}
			v_j^{|\tilde{\pi}^{2}_{j-1}| - |\tilde{\pi}^{2}_j|}	\notag\\
			=&
			\sum_{{\tilde{\pi}^{1} \subseteq [N_{1}, N_{1}, M_{1}]}
				\atop{\tilde{\pi}^{2} \subseteq [N_{2}, N_{2}, M_{2}]}}
			B_{\tilde{\pi}^{1} }(\{\mathbf{x^1}\}, \{\mathbf{z^1}\}) B_{\tilde{\pi}^{2} }(\{\mathbf{y^2}\}, \{\mathbf{v^2}\}),
		\end{align}
		where a sequence of interlacing strict boxed partitions in the summation symbol of Lemma $\ref{M-n-0}$  constitutes  strict boxed  plane partitions $\tilde{\pi}^{1}=(\ldots,\tilde{\pi}_{-1}^{1},\tilde{\pi}_{0}^{1},\tilde{\pi}_{1}^{1},\ldots)$ and  $\tilde{\pi}^{2}=(\ldots,\tilde{\pi}_{-1}^{2},\tilde{\pi}_{0}^{2},\tilde{\pi}_{1}^{2},\ldots)$  satisfying Eq.$(\ref{interlacing-partitions})$
		and
		\begin{align}
			B_{\tilde{\pi}^{1} }(\{\mathbf{x^1}\}, \{\mathbf{z^1}\})
			&=2^{p(\tilde{\pi}^1)}
			\prod_{i,l=1}^{N_1}
			x_i^{|\tilde{\pi}^{1}_{i-1}| - |\tilde{\pi}^{1}_i|}
			z_l^{|\tilde{\pi}^{2}_{l-1}| - |\tilde{\pi}^{2}_l|},\\
			B_{\tilde{\pi}^{2} }(\{\mathbf{y^2}\}, \{\mathbf{v^2}\})
			&=2^{p(\tilde{\pi}^2)}
			\prod_{j,k=1}^{N_2}
			y_j^{|\tilde{\pi}^{1}_{j-1}| - |\tilde{\pi}^{1}_j|}
			v_k^{|\tilde{\pi}^{2}_{k-1}| - |\tilde{\pi}^{2}_k|}	.
		\end{align}
		Note that $(\tilde{\chi}_{i})=(\tilde{\pi}_{i}^{1},\tilde{\pi}_{i}^{2})$,
		the scalar product generates   boxed BUC plane partition	 $\tilde{\Pi}^{'}=(\ldots,\tilde{\chi}_{-1},\tilde{\chi}_{0},\\ \tilde{\chi}_{1}, \ldots)$, which denotes a pair of strict boxed plane partitions
		${\tilde{\pi}^{1} \subseteq [N_{1}, N_{1}, M_{1}]}$ and
		$\tilde{\pi}^{2} \subseteq [N_{2}, N_{2}, M_{2}]$.
		
		Based on the Lemma $\ref{M-n-0-schur}$,  the scalar product can be expressed as
		\begin{align}\label{scalar-product-S1-S2}
			S(\{\mathbf{x^1}\}, \{\mathbf{y^2}\}, \{\mathbf{z^1}\}, \{\mathbf{v^2}\})	
			=&
			\sum_{{\tilde{\pi}^{1} \subseteq [N_{1}, N_{1}, M_{1}]}
				\atop{\tilde{\pi}^{2} \subseteq [N_{2}, N_{2}, M_{2}]}}
			B_{\tilde{\pi}^{1} }(\{\mathbf{x^1}\}, \{\mathbf{z^1}\}) B_{\tilde{\pi}^{2} }(\{\mathbf{y^2}\}, \{\mathbf{v^2}\})\notag\\
			=&\sum_{{ \tilde{\mu}^{1} \subseteq [N_1,M_1]}\atop{\tilde{\mu}^{2} \subseteq [N_2,M_2]}}
			2^{- l(\tilde{\mu}^{1})}	2^{- l(\tilde{\mu}^{2})}	
			Q_{\tilde{\mu}^{1}}\{\mathbf{x^1}\}  Q_{\tilde{\mu}^{1}}\{\mathbf{z^1}\}
			Q_{\tilde{\mu}^{2}}\{\mathbf{y^2}\}  Q_{\tilde{\mu}^{2}}\{\mathbf{v^2}\} ,
		\end{align}	
		which implies
		\begin{align}
			\sum_{\tilde{\pi}^{1} \subseteq [N_{1}, N_{1}, M_{1}]}		
			B_{\tilde{\pi}^{1} }(\{\mathbf{x^1}\}, \{\mathbf{z^1}\})
			&=
			\sum_{ \tilde{\mu}^{1} \subseteq [N_1,M_1]}
			2^{- l(\tilde{\mu}^{1})}
			Q_{\tilde{\mu}^{1}}\{\mathbf{x^1}\}  Q_{\tilde{\mu}^{1}}\{\mathbf{z^1}\}, \\
			\sum_{\tilde{\pi}^{2} \subseteq [N_{2}, N_{2}, M_{2}]}
			B_{\tilde{\pi}^{2} }(\{\mathbf{y^2}\}, \{\mathbf{v^2}\})
			&=\sum_{\tilde{\mu}^{2} \subseteq [N_2,M_2]}
			2^{- l(\tilde{\mu}^{2})}	
			Q_{\tilde{\mu}^{2}}\{\mathbf{y^2}\}
			Q_{\tilde{\mu}^{2}}\{\mathbf{v^2}\}.
		\end{align}
	\end{proof}
	It can be seen that the scalar product of the   generalized $i$-boson model leads to the generating function for  boxed BUC plane partitions which can
	be  represented as products of Schur $Q$-functions.

	\section{The generating function for boxed BUC plane partitions with the double scaling limit}
	
	This section discusses the generating function of boxed BUC plane partitions
	in certain limits.
	With the help of the neutral fermionic Fock derivation, the action of the monodromy matrix operators has been studied
	when the number of lattice sites $M_{j}\to \infty$.
	Furthermore, we obtain the generating function for BUC plane partitions with the double scaling limit, which  indicates
	that the number of lattice sites $M_{1},M_{2} \to \infty$ and the number of particles $N_{1},N_{2} \to \infty$.
	\subsection{Generating interlacing $2$-partitions}
	Define neutral fermion vertex operators
	\begin{align}\label{UCH}
		\tilde{\Gamma}_+(z,v) &= e^{\tilde{H}_{+}(z,v)}=
		\exp\left( \sum\limits_{n\in\mathbf{N}_{odd}}
		\left(\frac{2}{n}z^n \lambda_{n}^{'} + \frac{2}{n} v^n \tilde{\lambda}^{'}_{n} \right)\right), \\
		\tilde{\Gamma}_-(z,v) &=  e^{\tilde{H}_{-}(z,v)}=
		\exp\left(  \sum\limits_{n\in\mathbf{N}_{odd}}  \left(\frac{2}{n}z^n \lambda_{-n}^{'} + \frac{2}{n} v^n\tilde{\lambda}^{'}_{-n} \right) \right),
	\end{align}
	where
	\begin{align}
		\tilde{H}_{+}(z,v)&=\sum\limits_{n\in\mathbf{N}_{odd}}
		\left(\frac{2}{n}z^n \lambda_{n}^{'} + \frac{2}{n} v^n \tilde{\lambda}^{'}_{n} \right),\label{BUCH+}\\
		\tilde{H}_{-}(z,v)&= \sum\limits_{n\in\mathbf{N}_{odd}}  \left(\frac{2}{n}z^n \lambda_{-n}^{'} + \frac{2}{n} v^n\tilde{\lambda}^{'}_{-n} \right).\label{BUCH-}
	\end{align}
	Since
	\begin{align}
		[\tilde{H}_{+}(z,v),\tilde{H}_{-}(x,y)]
		&=4
		\sum\limits_{m,n\in\mathbf{N}_{odd}}   \left(\frac{z^{m}x^{n}}{mn}[\lambda_{m}^{'},\lambda_{-n}^{'}]
		+  \frac{v^{m}y^{n}}{mn}[\tilde{\lambda}^{'}_{m},\tilde{\lambda}^{'}_{-n}]\right) \notag\\
		&=4
		\sum\limits_{m,n\in\mathbf{N}_{odd}}
		\left(\frac{z^{m}x^{n}}{mn} \frac{m}{2} \delta_{m-n,0}
		+  \frac{v^{m}y^{n}}{mn}\frac{m}{2} \delta_{m-n,0}\right)
		\notag\\
		&=2\sum\limits_{m\in\mathbf{N}_{odd}}
		\left(\frac{(zx)^{m}}{m} + \frac{(vy)^{m}}{m}\right),
	\end{align}
	then
	\begin{align}\label{Gamme-commutation-relations}
		\tilde{\Gamma}_+(z,v) \tilde{\Gamma}_-(x,y)&=e^{[\tilde{H}_{+}(z,v),\tilde{H}_{-}(x,y)]} \tilde{\Gamma}_-(x,y)\tilde{\Gamma}_+(z,v)\notag\\
		&=\frac{1+xz}{1-xz} \frac{1+yv}{1-yv}   \tilde{\Gamma}_-(x,y)\tilde{\Gamma}_+(z,v).
	\end{align}

	\begin{lemma}\label{H1}
		For generating functions of neutral fermions
		\begin{align}\label{BUC-feimisc}	
			\phi(k){=}\sum_{n\in\mathbf{Z}}\phi_{n}k^{n},\quad \bar{\phi}(k){=}\sum_{n\in\mathbf{Z}}\bar{\phi}_{n}k^{n},
		\end{align}
		we have
		\begin{align}\label{commutation relations 3.1}
			[\tilde{H}_{+}(z,v),\phi(k)]&=\sum\limits_{n\in\mathbf{N}_{odd}}  \frac{2}{n} (zk)^n\phi(k),\quad
			[\tilde{H}_{-}(z,v), \phi(k)]
			=\sum\limits_{n\in\mathbf{N}_{odd}}  \frac{2}{n}  \left(\frac{z}{k}\right)^n \phi(k),\\
			[\tilde{H}_{+}(z,v),\bar{\phi}(k)]&=\sum\limits_{n\in\mathbf{N}_{odd}}  \frac{2}{n} (vk)^n \bar{\phi}(k),\quad
			[\tilde{H}_{-}(z,v),\bar{\phi}(k)]=\sum\limits_{n\in\mathbf{N}_{odd}} \frac{2}{n}   \left(\frac{v}{k}\right)^n \bar{\phi}(k),
		\end{align}
		and	the following equations hold
		\begin{align}\label{commutation relations 4.1}	
			\tilde{\Gamma}_+(z,v)\phi(k)\tilde{\Gamma}_+^{-1}(z,v)
			&=\frac{1+zk}{1-zk}\phi(k),\\		\tilde{\Gamma}_-(z,v)\phi(k)\tilde{\Gamma}_-^{-1}(z,v)
			&=\frac{1+\frac{z}{k}}{1-\frac{z}{k}}  \phi(k),\\ \tilde{\Gamma}_+(z,v)\bar{\phi}(k)\tilde{\Gamma}_+^{-1}(z,v)
			&=\frac{1+vk}{1-vk}\bar{\phi}(k),\\
			\tilde{\Gamma}_-(z,v)\bar{\phi}(k)\tilde{\Gamma}_-^{-1}(z,v)
			&=\frac{1+\frac{v}{k}}{1-\frac{v}{k}} \bar{\phi}(k).
		\end{align}		
	\end{lemma}
	\begin{proof}
		Combining Eqs.(\ref{commutation relations 1.2}) and (\ref{BUCH+}) yields
		\begin{align}
			[\tilde{H}_{+}(z,v),\phi(k)]
			&=\sum_{n=1}^{\infty}\sum_{m\in\mathbf{Z}+1/2}
			\frac{2}{n}(zk)^n  \phi_{n-m} k^{n-m}		\notag\\
			&= \sum_{n=1}^{\infty} \frac{2}{n}(zk)^n   \phi(k).
		\end{align}
		It can be derived	from Eq.(\ref{commutation relations 3.1}) that
		\begin{align}\label{UC-TD}
			\tilde{\Gamma}_+(z,v)\phi(k)\tilde{\Gamma}_+^{-1}(z,v)
			&=e^{\tilde{H}_{+}(z,v)}	\phi(k) e^{-\tilde{H}_{+}(z,v)}
			\notag\\	
			 &=\phi(k)+[\tilde{H}_{+}(z,v),\phi(k)]+\frac{1}{2!}[\tilde{H}_{+}(z,v),[\tilde{H}_{+}(z,v),\phi(k)]]+\ldots
			\notag\\	
			&=\exp{\left(\sum\limits_{n=1}^{\infty} \frac{2}{n}(zk)^n \right) }\phi(k).
		\end{align}
		Due to $\sum\limits_{n=1}^{\infty} \frac{2}{n}(zk)^n=-ln(\frac{1+zk}{1-zk})$, we obtain
		\begin{align}
			\tilde{\Gamma}_+(z,v)\phi(k)\tilde{\Gamma}_+^{-1}(z,v)
			=\frac{1+zk}{1-zk}	\phi(k).
		\end{align}
		
		Using the similar process, other equations can be proved.		
	\end{proof}
	
	\begin{proposition}
		The vertex operators and neutral fermions satisfy
		\begin{align}
			&\tilde{\Gamma}_+(z,v)\phi_{i}
			=\left( \phi_i + 2\sum_{n=1}^{\infty} \phi_{(i-n)} z^n \right)
			\tilde{\Gamma}_+(z,v),\\		
			&\tilde{\Gamma}_+(z,v)\bar{\phi}_{i}
			=\left( \bar{\phi}_i + 2\sum_{n=1}^{\infty} \bar{\phi}_{(i-n)} v^n \right)
			\tilde{\Gamma}_+(z,v),	\\
			&\tilde{\Gamma}_-(z,v)\left( \phi_i^* + 2\sum_{n=1}^{\infty}  \phi_{(i-n)}^*  z^n \right)	
			=\phi_{i}^*\tilde{\Gamma}_-(z,v),\label{Gamma-phi}	\\
			&\tilde{\Gamma}_-(z,v)\left( \bar{\phi}_i^* + 2\sum_{n=1}^{\infty}  \bar{\phi}_{(i-n)}^*  v^n \right)	
			=\bar{\phi}_{i}^*\tilde{\Gamma}_-(z,v).
		\end{align}
	\end{proposition}
	\begin{proof}
		By using the Lemma $\ref{H1}$ and the identity
		\begin{align}
			\frac{1+t}{1-t}=1+2(\frac{1}{1-t}-1)
			=1+2\sum_{n=1}^{\infty}t^n,
		\end{align}	
		one gets	
		\begin{align}
			\tilde{\Gamma}_+(z,v)
			\sum_{i\in\mathbf{Z}}\phi_{i}k^{i}&=
			\left(1+2\sum_{n=1}^{\infty}(zk)^n\right)
			\sum_{j\in\mathbf{Z}}\phi_{j}k^{j} \tilde{\Gamma}_+(z,v)\notag\\
			&=\left(\sum_{j\in\mathbf{Z}}\phi_{j}k^{j}
			+2\sum_{n=1}^{\infty}\sum_{j\in\mathbf{Z}}  z^n k^{j+n} \phi_{j}
			\right)
			\tilde{\Gamma}_+(z,v).
		\end{align}
		Comparing the orders of $k$ on both sides, it is obvious that
		\begin{align}
			\tilde{\Gamma}_+(z,v)\phi_{i}
			=\left( \phi_i + 2\sum_{n=1}^{\infty} \phi_{(i-n)} z^n \right)
			\tilde{\Gamma}_+(z,v).
		\end{align}
		
		Similarly, based on the identity
		\begin{align}
			\frac{1-t}{1+t}=1+2\left(\frac{1}{1+t}-1\right)
			=1+2\sum_{n=1}^{\infty} (-1)^n t^n,
		\end{align}
		Eq.($\ref{Gamma-phi}$)  is rewritten as
		\begin{align}
			\tilde{\Gamma}_-(z,v)
			\left(\sum_{j\in\mathbf{Z}}\phi_{j}k^{j}
			+2\sum_{n=1}^{\infty}\sum_{j\in\mathbf{Z}}(-1)^n  z^n k^{j-n} \phi_{j}
			\right)
			=\sum_{i\in\mathbf{Z}}\phi_{i}k^{i}  \tilde{\Gamma}_-(z,v),
		\end{align}
		which means that
		\begin{align}
			\tilde{\Gamma}_-(z,v)\left( \phi_i + 2\sum_{n=1}^{\infty}  \phi_{(n+i)} (-1)^{n} z^n \right)	
			=\phi_{i}\tilde{\Gamma}_-(z,v).
		\end{align}
		Set $i\rightarrow -i$ and multiply both sides of the equation by $(-1)^i$ to deduce
		\begin{align}
			\tilde{\Gamma}_-(z,v)\left((-1)^i \phi_{-i} + 2\sum_{n=1}^{\infty}  (-1)^{i-n} \phi_{(n-i)}  z^n \right)	
			=(-1)^i\phi_{-i}\tilde{\Gamma}_-(z,v).
		\end{align}
		Thus
		\begin{align}
			\tilde{\Gamma}_-(z,v)\left( \phi_i^* + 2\sum_{n=1}^{\infty}  \phi_{(i-n)}^*  z^n \right)	
			=\phi_{i}^*\tilde{\Gamma}_-(z,v).
		\end{align}
		Other formulas can be proved by the same method.
	\end{proof}

	\begin{lemma}\label{Gamma-state}
		The actions of vertex operators $\tilde{\Gamma}_+(z,v)$  and $\tilde{\Gamma}_-(z,v)$ on state vectors satisfy
		\begin{align}
			\langle \tilde{\mu}^2,\tilde{\mu}^1|\tilde{\Gamma}_-(z,v) &=
			\sum_{{\tilde{\nu}^1 \prec \tilde{\mu}^1}	\atop{\tilde{\nu}^2 \prec \tilde{\mu}^2}}
			2^{\#(\tilde{\mu}^1|\tilde{\nu}^1) }2^{\#(\tilde{\mu}^2|\tilde{\nu}^2) }
			z^{|\tilde{\mu}^1| - |\tilde{\nu}^1|}    v^{|\tilde{\mu}^2| - |\tilde{\nu}^2|}     \langle \tilde{\nu}^2,\tilde{\nu}^1|,\\
			\tilde{\Gamma}_+(z,v)|\tilde{\mu}^1,\tilde{\mu}^2 \rangle &= 	
			\sum_{{\tilde{\nu}^1 \prec \tilde{\mu}^1}	\atop{\tilde{\nu}^2 \prec \tilde{\mu}^2}}
			2^{\#(\tilde{\mu}^1|\tilde{\nu}^1) }2^{\#(\tilde{\mu}^2|\tilde{\nu}^2) }
			z^{|\tilde{\mu}^1| - |\tilde{\nu}^1|}    v^{|\tilde{\mu}^2| - |\tilde{\nu}^2|}         |\tilde{\nu}^1,\tilde{\nu}^2 \rangle.
		\end{align}
	\end{lemma}
	\begin{proof}
		Combining the  fact
		$\langle 0| \tilde{\Gamma}_-(z,v)=\langle 0| $
		and Eq.$(\ref{o-left-state})$, one can get
		\begin{align}\label{Gamma_+-1}
			\langle \tilde{\mu}^2,\tilde{\mu}^1|\tilde{\Gamma}_-(z,v)
			= \langle 0|
			&\bar{\phi}_{\tilde{\mu}^{2}_{2r_{2}}}^* \dots \bar{\phi}_{\tilde{\mu}^{2}_1}^*
			\phi_{\tilde{\mu}^{1}_{2r_{1}}}^* \dots \phi_{\tilde{\mu}^{1}_1}^*
			\tilde{\Gamma}_-(z,v)\notag\\
			=\langle 0|
			&\left( \bar{\phi}_{\tilde{\mu}^{2}_{2r_{2}}}^* + 2\sum_{j_{2r_{2}}=1}^{\infty}
			\bar{\phi}_{(\tilde{\mu}^{2}_{2r_{2}}-j_{2r_{2}})}^*  v^{j_{2r_{2}}} \right)
			\ldots
			\left( \bar{\phi}_{\tilde{\mu}^{2}_1}^* + 2\sum_{j_1=1}^{\infty}
			\bar{\phi}_{(\tilde{\mu}^{2}_1-j_1)}^*  v^{j_1} \right)\notag\\
			&\left( \phi_{\tilde{\mu}^{1}_{2r_{1}}}^* + 2\sum_{i_{2r_{1}}=1}^{\infty}
			\phi_{(\tilde{\mu}^{1}_{2r_{1}}-i_{2r_{1}})}^*  z^{i_{2r_{1}}} \right)
			\ldots
			\left( \phi_{\tilde{\mu}^{1}_1}^* + 2\sum_{i_1=1}^{\infty}
			\phi_{(\tilde{\mu}^{1}_1-i_1)}^*  z^{i_1} \right)	\notag\\
			=\langle 0|
			&\left( \sum_{j_{2r_{2}}=0}^{\infty}  (2 - \delta_{j_{2r_{2}},0})
			\bar{\phi}_{(\tilde{\mu}^{2}_{2r_{2}}-j_{2r_{2}})}^*  v^{j_{2r_{2}}} \right)
			\ldots
			\left( \sum_{j_1=0}^{\infty}  (2 - \delta_{j_{1},0})
			\bar{\phi}_{(\tilde{\mu}^{2}_1-j_1)}^*  v^{j_1} \right)\notag\\
			&\left( \sum_{i_{2r_{1}}=0}^{\infty}  (2 - \delta_{i_{2r_{1}},0})
			\phi_{(\tilde{\mu}^{1}_{2r_{1}}-i_{2r_{1}})}^*  z^{i_{2r_{1}}} \right)
			\ldots
			\left( \sum_{i_1=0}^{\infty}  (2 - \delta_{i_{1},0})
			\phi_{(\tilde{\mu}^{1}_1-i_1)}^*  z^{i_1} \right)	.
		\end{align}
		By using  identities
		\begin{align}
			\left( \sum_{n=0}^{\infty} (2 - \delta_{i,0}) z^n \phi_{(i-n)}^* \right)
			\left( \sum_{m=0}^{\infty} (2 - \delta_{j,0}) z^n \phi_{(j-m)}^* \right) = 0,\\
			\left( \sum_{n=0}^{\infty} (2 - \delta_{i,0}) v^n \bar{\phi}_{(i-n)}^* \right)
			\left( \sum_{m=0}^{\infty} (2 - \delta_{j,0}) v^n \bar{\phi}_{(j-m)}^* \right) = 0,
		\end{align}
		Eq.$(\ref{Gamma_+-1})$ is rewritten as
		\begin{align}
			\langle \tilde{\mu}^2,\tilde{\mu}^1|\tilde{\Gamma}_-(z,v)
			=\langle 0|T_{2r_{2}}^{2}(v)T_{2r_{2}-1}^{2}(v) \cdots T_{1}^{2}(v)
			\cdot	T_{2r_{1}}^{1}(z)T_{2r_{1}-1}^{1}(z) \cdots T_{1}^{1}(z),
		\end{align}
		where
		\begin{align}
			T_{l}^{1}(z)=&\left\{
			\begin{array}
				{ll}\sum\limits_{i_l=0}^{\tilde{\mu}^{1}_l-\tilde{\mu}^{1}_{l+1}}
				(2 - \delta_{i_{l},0}-\delta_{i_{l},\tilde{\mu}^{1}_{l}-\tilde{\mu}^{1}_{l+1}})
				\phi_{(\tilde{\mu}^{1}_{l}-i_{l})}^*  z^{i_{l}} , & 1 \leq l \leq 2r_{1}-1, \\
				\\
				\sum\limits_{i_{2r_{1}}=0}^{\tilde{\mu}^{1}_{2r_{1}}}  (2 - \delta_{i_{2r_{1}},0})
				\phi_{(\tilde{\mu}^{1}_{2r_{1}}-i_{2r_{1}})}^*  z^{i_{2r_{1}}},
				& l=2r_{1},
			\end{array}\right.\\
			T_{k}^{2}(v)=&\left\{
			\begin{array}
				{ll}\sum\limits_{i_l=0}^{\tilde{\mu}^{2}_k-\tilde{\mu}^{2}_{k+1}}
				(2 - \delta_{i_{k},0}-\delta_{i_{k},\tilde{\mu}^{2}_{k}-\tilde{\mu}^{2}_{k+1}})
				\phi_{(\tilde{\mu}^{2}_{k}-i_{k})}^*  v^{i_{k}} , & 1 \leq k \leq 2r_{2}-1, \\
				\\
				\sum\limits_{i_{2r_{2}}=0}^{\tilde{\mu}^{2}_{2r_{2}}}  (2 - \delta_{i_{2r_{2}},0})
				\phi_{(\tilde{\mu}^{2}_{2r_{2}}-i_{2r_{2}})}^*  v^{i_{2r_{2}}},
				& k=2r_{2}.
			\end{array}\right.
		\end{align}

		For all $j=1,2$ and $1 \leq t \leq 2r_{j}$, setting $\tilde{\mu}^{j}_{t}-i_{t}=\tilde{\nu}^{j}_{t}$,
		then
		\begin{align}
			T_{l}^{1}(z)=&\left\{
			\begin{array}
				{ll}\sum\limits_{\tilde{\mu}^{1}_{l} \geq  \tilde{\nu}^{1}_{l} \geq \tilde{\mu}^{1}_{l+1}}
				(2 - \delta_{\tilde{\mu}^{1}_{l},\tilde{\nu}^{1}_{l}}
				-\delta_{\tilde{\mu}^{1}_{l+1},\tilde{\nu}^{1}_{l}}) 	
				\phi_{\tilde{\nu}^{1}_{l}}^*  z^{\tilde{\mu}^{1}_{l}-\tilde{\nu}^{2}_{l}} , & 1 \leq l \leq 2r_{1}-1, \\
				\\
				\sum\limits_{\tilde{\mu}^{1}_{2r_{1}} \geq  \tilde{\nu}^{1}_{2r_{1}} \geq 0}  (2 - \delta_{\tilde{\mu}^{1}_{2r_{1}},\tilde{\nu}^{1}_{2r_{1}}})
				\phi_{\tilde{\nu}^{1}_{2r_{1}}}^*  z^{\tilde{\mu}^{1}_{2r_{1}}-\tilde{\nu}^{1}_{2r_{1}}},
				& l=2r_{1},
			\end{array}\right.\\
			T_{k}^{2}(v)=&\left\{
			\begin{array}
				{ll}\sum\limits_{\tilde{\mu}^{2}_{k} \geq  \tilde{\nu}^{2}_{k} \geq \tilde{\mu}^{2}_{k+1}}
				(2 - \delta_{\tilde{\mu}^{2}_{k},\tilde{\nu}^{2}_{k}}
				-\delta_{\tilde{\mu}^{2}_{k+1},\tilde{\nu}^{2}_{k}}) 	
				\bar{\phi}_{\tilde{\nu}^{2}_{k}}^*  v^{\tilde{\mu}^{2}_{k}-\tilde{\nu}^{2}_{k}}, & 1 \leq k \leq 2r_{2}-1, \\
				\\
				\sum\limits_{\tilde{\mu}^{2}_{2r_{2}} \geq  \tilde{\nu}^{2}_{2r_{2}} \geq 0}  (2 - \delta_{\tilde{\mu}^{2}_{2r_{2}},\tilde{\nu}^{2}_{2r_{2}}})
				\bar{\phi}_{\tilde{\nu}^{2}_{2r_{2}}}^*  v^{\tilde{\mu}^{2}_{2r_{2}}-\tilde{\nu}^{2}_{2r_{2}}},
				& k=2r_{2}.
			\end{array}\right.
		\end{align}
		Note
		\begin{align}
			\langle \tilde{\nu}^2,\tilde{\nu}^1|=\langle 0|\bar{\phi}_{\tilde{\nu}^{2}_{2r_{2}}}^* \dots \bar{\phi}_{\tilde{\nu}^{2}_1}^*
			\phi_{\tilde{\nu}^{1}_{2r_{1}}}^* \dots \phi_{\tilde{\nu}^{1}_1}^*.
		\end{align}
From the following subscript relations,
		\begin{align}
			&m_{i}^{j} \geq n^j_{i} \geq m_{i+1}^{j}  \implies  \tilde{\mu}_i^{j}\geq \tilde{\nu}_i^{j} \geq \tilde{\mu}_{i+1}^{j},\quad
			\text{ for all } 1 \leq i \leq 2r_{j}-1, \notag\\
			&m_{i}^{j} \geq n^j_{i} \geq 0
			\implies  \tilde{\mu}_i^{j} \geq \tilde{\nu}_i^{j} \geq 0
			\text{ for } i=2r_{j}, \;  j=1,2,
		\end{align}
we have
		\begin{align}
			\langle \tilde{\mu}^2,\tilde{\mu}^1|\tilde{\Gamma}_-(z,v)
			=\sum_{{\tilde{\nu}^1 \prec \tilde{\mu}^1}	\atop{\tilde{\nu}^2 \prec \tilde{\mu}^2}}
			2^{\#(\tilde{\mu}^1|\tilde{\nu}^1) }2^{\#(\tilde{\mu}^2|\tilde{\nu}^2) }
			z^{|\tilde{\mu}^1| - |\tilde{\nu}^1|}    v^{|\tilde{\mu}^2| - |\tilde{\nu}^2|}     \langle \tilde{\nu}^2,\tilde{\nu}^1|.
		\end{align}

		Analogous to the above derivation, we obtain
		\begin{align}
			\tilde{\Gamma}_+(z,v)|\tilde{\mu}^1,\tilde{\mu}^2 \rangle &= 	
			\sum_{{\tilde{\nu}^1 \prec \tilde{\mu}^1}	\atop{\tilde{\nu}^2 \prec \tilde{\mu}^2}}
			2^{\#(\tilde{\mu}^1|\tilde{\nu}^1) }2^{\#(\tilde{\mu}^2|\tilde{\nu}^2) }
			z^{|\tilde{\mu}^1| - |\tilde{\nu}^1|}    v^{|\tilde{\mu}^2| - |\tilde{\nu}^2|}         |\tilde{\nu}^1,\tilde{\nu}^2 \rangle.
		\end{align}
	\end{proof}
	\begin{remark}\label{Gamma+mu1,mu1-Gamma}
		For $j=1,2$, the case of  one of the partition being $\emptyset$ in Lemma $\ref{Gamma-state}$ can be inferred from the above derivation that
		\begin{align}
			\langle \tilde{\mu}^j|\tilde{\Gamma}_-(z_1,z_2) &=
			\sum_{\tilde{\nu}^j \prec \tilde{\mu}^j}
			2^{\#(\tilde{\mu}^j|\tilde{\nu}^j) } z_{j}^{|\tilde{\mu}^j| - |\tilde{\nu}^j|}    \langle \tilde{\nu}^j|,\\
			\tilde{\Gamma}_+(z_1,z_2)|\tilde{\mu}^j \rangle &=
			\sum_{\tilde{\nu}^j \prec \tilde{\mu}^j}
			2^{\#(\tilde{\mu}^j|\tilde{\nu}^j) } z_{j}^{|\tilde{\mu}^j| - |\tilde{\nu}^j|}           |\tilde{\nu}^j \rangle	.
		\end{align}
		
	\end{remark}

	\subsection{The limiting forms of the scalar product of the   generalized $i$-boson model}
	Consider the infinite lattice limit $M_{1},M_{2} \to \infty$ for the   generalized $i$-boson model.
	\begin{lemma}\label{M-wq-n}
		For basis vectors $|\tilde{n}^{1}\rangle^{(1)}\bigotimes|\tilde{n}^{2}\rangle^{(2)}$ and $	 \prescript{(2)}{}{\langle \tilde{n}^{2}|}\bigotimes \prescript{(1)}{}{\langle \tilde{n}^{1}|} $,
		set
		\begin{align}
			\tilde{\mathcal{M}}  | \tilde{n}^{1}\rangle^{(1)}\bigotimes|\tilde{n}^{2}\rangle^{(2)}&=
			2^{-l(\tilde{\nu}^{1})-l(\tilde{\nu}^{2})}	
			| \tilde{\nu}^{1},\tilde{\nu}^{2} \rangle, \\	
			\prescript{(2)}{}{\langle \tilde{n}^{2}|}\bigotimes \prescript{(1)}{}{\langle \tilde{n}^{1}|}  \tilde{\mathcal{M}}^* &=
			2^{-l(\tilde{\nu}^{1})-l(\tilde{\nu}^{2})}	
			\langle \tilde{\nu}^{2},\tilde{\nu}^{1}|.
		\end{align}	
		The following equations hold
		\begin{align}
			\tilde{\mathcal{M}} \left[ \lim_{M_{1},M_{2} \to \infty} \tilde{\mathbb{B}}_{1}(z) \tilde{\mathbb{B}}_{2}(v)|\tilde{n}^{1}\rangle^{(1)}\bigotimes|\tilde{n}^{2}\rangle^{(2)}\right] &=
			2^{- l(\tilde{\nu}^{1})}	2^{- l(\tilde{\nu}^{2})}	
			\tilde{\Gamma}_-(z,v) | \tilde{\nu}^1,\tilde{\nu}^2 \rangle, \\
			\quad \left[ \lim_{M_{1},M_{2}  \to \infty}
			\prescript{(2)}{}{\langle \tilde{n}^{2}|}\bigotimes \prescript{(1)}{}{\langle \tilde{n}^{1}|}  \tilde{\mathbb{C}}_{2}(v)\tilde{\mathbb{C}}_{1}(z) \right] \tilde{\mathcal{M}}^* &=
			2^{- l(\tilde{\nu}^{1})}	2^{- l(\tilde{\nu}^{2})}	
			\langle \tilde{\nu}^2,\tilde{\nu}^1| \tilde{\Gamma}_+(z,v).
		\end{align}
	\end{lemma}
	\begin{proof}
		On account of  Eqs.$(\ref{M-B1-B2-N})$ and $(\ref{M-C1-C2-N})$, we get
		\begin{align}
			&\tilde{\mathcal{M}} \left[ \lim_{M_{1},M_{2} \to \infty} \mathbb{B}_{1}(z) \mathbb{B}_{2}(v)|\tilde{n}^{1}\rangle^{(1)}\bigotimes|\tilde{n}^{2}\rangle^{(2)} \right] \notag\\
			=&\sum_{{\tilde{\nu}^1 \prec \tilde{\mu}^1}	\atop{\tilde{\nu}^2 \prec \tilde{\mu}^2}}
			2^{\#(\tilde{\mu}^{1}|\tilde{\nu}^{1}) - l(\tilde{\mu}^{1})}
			2^{\#(\tilde{\mu}^{2}|\tilde{\nu}^{2}) - l(\tilde{\mu}^{2})}
			z^{|\tilde{\mu}^{1}| - |\tilde{\nu}^{1}|} v^{|\tilde{\mu}^{2}| - |\tilde{\nu}^{2}|}
			| \tilde{\mu}^{1},\tilde{\mu}^{2} \rangle,  \\
			&\left[ \lim_{M_{1},M_{2}  \to \infty}
			\prescript{(2)}{}{\langle \tilde{n}^{2}|}\bigotimes \prescript{(1)}{}{\langle \tilde{n}^{1}|}
			\mathbb{C}_{2}(v)\mathbb{C}_{1}(z) \right] \tilde{\mathcal{M}}^* \notag\\
			=&\sum_{{\tilde{\nu}^1 \prec \tilde{\mu}^1}	\atop{\tilde{\nu}^2 \prec \tilde{\mu}^2}}
			2^{\#(\tilde{\mu}^{1}|\tilde{\nu}^{1}) - l(\tilde{\mu}^{1})}
			2^{\#(\tilde{\mu}^{2}|\tilde{\nu}^{2}) - l(\tilde{\mu}^{2})}
			z^{|\tilde{\mu}^{1}| - |\tilde{\nu}^{1}|} v^{|\tilde{\mu}^{2}| - |\tilde{\nu}^{2}|}
			\langle \tilde{\mu}^{2},\tilde{\mu}^{1} |.
		\end{align}
		Then
		\begin{align}
			\tilde{\Gamma}_-(z,v) |\tilde{\nu}^1,\tilde{\nu}^2 \rangle
			=&\sum_{{\tilde{\nu}^1 \prec \tilde{\mu}^1}	\atop{\tilde{\nu}^2 \prec \tilde{\mu}^2}}
			2^{\#(\tilde{\mu}^{1}|\tilde{\nu}^{1}) - l(\tilde{\mu}^{1})+l(\tilde{\nu}^{1})}
			2^{\#(\tilde{\mu}^{2}|\tilde{\nu}^{2}) - l(\tilde{\mu}^{2})+l(\tilde{\nu}^{2})}
			z^{|\tilde{\mu}^{1}| - |\tilde{\nu}^{1}|} v^{|\tilde{\mu}^{2}| - |\tilde{\nu}^{2}|}
			| \tilde{\mu}^{1},\tilde{\mu}^{2} \rangle,  \label{Gamma--right}  \\
			\langle \tilde{\nu}^2,\tilde{\nu}^1|\tilde{\Gamma}_+(z,v)
			=&\sum_{{\tilde{\nu}^1 \prec \tilde{\mu}^1}	\atop{\tilde{\nu}^2 \prec \tilde{\mu}^2}}
			2^{\#(\tilde{\mu}^{1}|\tilde{\nu}^{1}) - l(\tilde{\mu}^{1})+l(\tilde{\nu}^{1})}
			2^{\#(\tilde{\mu}^{2}|\tilde{\nu}^{2}) - l(\tilde{\mu}^{2})+l(\tilde{\nu}^{2})}
			z^{|\tilde{\mu}^{1}| - |\tilde{\nu}^{1}|} v^{|\tilde{\mu}^{2}| - |\tilde{\nu}^{2}|}
			\langle \tilde{\nu}^2,\tilde{\nu}^1|.\label{Gamma+-left}
		\end{align}
		It follows from the orthogonality of Eq.$(\ref{fock-scalar-product})$ that
		\begin{align}
			\langle \tilde{\mu}^2,\tilde{\mu}^1|\tilde{\Gamma}_-(z,v) |\tilde{\nu}^1,\tilde{\nu}^2 \rangle
			&=
			\langle \tilde{\nu}^2,\tilde{\nu}^1|\tilde{\Gamma}_+(z,v)|\tilde{\mu}^1,\tilde{\mu}^2 \rangle\notag\\
			&=	
			2^{\#(\tilde{\mu}^{1}|\tilde{\nu}^{1})+l(\tilde{\nu}^{1})}
			2^{\#(\tilde{\mu}^{2}|\tilde{\nu}^{2}) +l(\tilde{\nu}^{2})}
			z^{|\tilde{\mu}^{1}| - |\tilde{\nu}^{1}|} v^{|\tilde{\mu}^{2}| - |\tilde{\nu}^{2}|}.
		\end{align}
		Assume that
		\begin{align}
			\langle \tilde{\mu}^2,\tilde{\mu}^1|\tilde{\Gamma}_-(z,v)
			=&\sum_{{\tilde{\nu}^1 \prec \tilde{\mu}^1}	\atop{\tilde{\nu}^2 \prec \tilde{\mu}^2}}
			2^{t}2^{\#(\tilde{\mu}^{1}|\tilde{\nu}^{1})}
			2^{\#(\tilde{\mu}^{2}|\tilde{\nu}^{2}) }
			z^{|\tilde{\mu}^{1}| - |\tilde{\nu}^{1}|} v^{|\tilde{\mu}^{2}| - |\tilde{\nu}^{2}|}\langle \tilde{\nu}^2,\tilde{\nu}^1|,\\
			\tilde{\Gamma}_+(z,v)|\tilde{\mu}^1,\tilde{\mu}^2 \rangle
			=&\sum_{{\tilde{\nu}^1 \prec \tilde{\mu}^1}	\atop{\tilde{\nu}^2 \prec \tilde{\mu}^2}}
			2^{t}2^{\#(\tilde{\mu}^{1}|\tilde{\nu}^{1})}
			2^{\#(\tilde{\mu}^{2}|\tilde{\nu}^{2}) }
			z^{|\tilde{\mu}^{1}| - |\tilde{\nu}^{1}|} v^{|\tilde{\mu}^{2}| - |\tilde{\nu}^{2}|}|\tilde{\nu}^1,\tilde{\nu}^2 \rangle .
		\end{align}
		Based on Eq.(\ref{fock-scalar-product}), we have
		\begin{align}
			\langle \tilde{\mu}^2,\tilde{\mu}^1|\tilde{\Gamma}_-(z,v) |\tilde{\nu}^1,\tilde{\nu}^2 \rangle
			=&\langle \tilde{\nu}^2,\tilde{\nu}^1|\tilde{\Gamma}_+(z,v)|\tilde{\mu}^1,\tilde{\mu}^2 \rangle\notag\\
			=&
			2^{t} 2^{l(\tilde{\nu}^{1})+l(\tilde{\nu}^{2})}
			2^{\#(\tilde{\mu}^{1}|\tilde{\nu}^{1})}
			2^{\#(\tilde{\mu}^{2}|\tilde{\nu}^{2}) }
			z^{|\tilde{\mu}^{1}| - |\tilde{\nu}^{1}|} v^{|\tilde{\mu}^{2}| - |\tilde{\nu}^{2}|}\notag\\
			=&
			2^{\#(\tilde{\mu}^{1}|\tilde{\nu}^{1})+l(\tilde{\nu}^{1})}
			2^{\#(\tilde{\mu}^{2}|\tilde{\nu}^{2}) +l(\tilde{\nu}^{2})}
			z^{|\tilde{\mu}^{1}| - |\tilde{\nu}^{1}|} v^{|\tilde{\mu}^{2}| - |\tilde{\nu}^{2}|}.
		\end{align}
		Thus $t=0$, than the following equations hold
		\begin{align}
			\langle \tilde{\mu}^2,\tilde{\mu}^1|\tilde{\Gamma}_-(z,v)
			=&\sum_{{\tilde{\nu}^1 \prec \tilde{\mu}^1}	\atop{\tilde{\nu}^2 \prec \tilde{\mu}^2}}
			2^{\#(\tilde{\mu}^{1}|\tilde{\nu}^{1})}
			2^{\#(\tilde{\mu}^{2}|\tilde{\nu}^{2}) }
			z^{|\tilde{\mu}^{1}| - |\tilde{\nu}^{1}|} v^{|\tilde{\mu}^{2}| - |\tilde{\nu}^{2}|}\langle \tilde{\nu}^2,\tilde{\nu}^1|, \label{Gamma-left}\\
			\tilde{\Gamma}_+(z,v)|\tilde{\mu}^1,\tilde{\mu}^2 \rangle
			=&\sum_{{\tilde{\nu}^1 \prec \tilde{\mu}^1}	\atop{\tilde{\nu}^2 \prec \tilde{\mu}^2}}
			2^{\#(\tilde{\mu}^{1}|\tilde{\nu}^{1})}
			2^{\#(\tilde{\mu}^{2}|\tilde{\nu}^{2}) }
			z^{|\tilde{\mu}^{1}| - |\tilde{\nu}^{1}|} v^{|\tilde{\mu}^{2}| - |\tilde{\nu}^{2}|}|\tilde{\nu}^1,\tilde{\nu}^2 \rangle .\label{Gamma+right}
		\end{align}
From Lemma $\ref{Gamma-state}$ that the proof is complete.
	\end{proof}
Similarly, for $j=1,2$, we have	
	\begin{align}
		\tilde{\Gamma}_-(z_{1},z_{2}) |\tilde{\nu}^j \rangle
		=&\sum_{\tilde{\nu}^j \prec \tilde{\mu}^j	}
		2^{\#(\tilde{\mu}^{j}|\tilde{\nu}^{j}) - l(\tilde{\mu}^{j})+l(\tilde{\nu}^{j})}
		z_{j}^{|\tilde{\mu}^{j}| - |\tilde{\nu}^{j}|}
		| \tilde{\mu}^{j}\rangle,   \label{Gamma-right-1}\\
		\langle \tilde{\nu}^j|\tilde{\Gamma}_+(z_{1},z_{2})
		=&\sum_{\tilde{\nu}^j \prec \tilde{\mu}^j	}
		2^{\#(\tilde{\mu}^{j}|\tilde{\nu}^{j}) - l(\tilde{\mu}^{j})+l(\tilde{\nu}^{j})}
		z_{j}^{|\tilde{\mu}^{j}| - |\tilde{\nu}^{j}|}
		\langle \tilde{\nu}^j|.\label{Gamma+left-1}
	\end{align}

	When the number of lattice site $M_{1},M_{2}$ is sufficiently large, the scalar product will generate boxed  BUC plane partition with unrestricted column heights as follows
	\begin{align}
		&\lim\limits_{M_{1},M_{2}\to \infty}	
		\tilde{S}(\{\mathbf{x^1}\}, \{\mathbf{y^2}\}, \{\mathbf{z^1}\}, \{\mathbf{v^2}\})\notag\\
		=&\lim\limits_{M_{1},M_{2}\to \infty}
		\bigg \langle \prescript{(2)}{}{\langle 0|}\bigotimes \prescript{(1)}{}{\langle 0|}
		\rprod_{j=1}^{N_2} \tilde{\mathbb{C}}_{2}(y_j)
		\prod_{i=1}^{N_1} \tilde{\mathbb{C}}_{1}(x_i) \tilde{\mathcal{M}}^*
		,\tilde{\mathcal{M}}\prod_{l=1}^{N_1} \tilde{\mathbb{B}}_{1}(z_i)
		\prod_{k=1}^{N_2} \tilde{\mathbb{B}}_{2}(v_k)
		|0\rangle^{(1)}\bigotimes|0\rangle^{(2)} \bigg\rangle\notag\\
		=&
		\sum_{{\tilde{\pi}^{1} \subseteq [N_{1}, N_{1},\infty]}
			\atop{\tilde{\pi}^{2} \subseteq [N_{2}, N_{2}, \infty]}}
		B_{\tilde{\pi}^{1} }(\{\mathbf{x^1}\}, \{\mathbf{z^1}\}) B_{\tilde{\pi}^{2} }(\{\mathbf{y^2}\}, \{\mathbf{v^2}\}).
	\end{align}
	Setting $N=\max\{N_{1},N_{2}\}$,
	it can be inferred from
	Eq.$(\ref{Gamme-commutation-relations})$ and the Lemma $\ref{M-wq-n}$ that
	\begin{align}\label{scalar-product-Gamma1-Gamma2-m-wq}
		&\lim\limits_{M_{1},M_{2}\to \infty}
		\bigg \langle \prescript{(2)}{}{\langle 0|}\bigotimes \prescript{(1)}{}{\langle 0|}
		\rprod_{j=1}^{N_2} \tilde{\mathbb{C}}_{2}(y_j)
		\prod_{i=1}^{N_1} \tilde{\mathbb{C}}_{1}(x_i) \tilde{\mathcal{M}}^*
		,\tilde{\mathcal{M}}\prod_{l=1}^{N_1} \tilde{\mathbb{B}}_{1}(z_i)
		\prod_{k=1}^{N_2} \tilde{\mathbb{B}}_{2}(v_k)
		|0\rangle^{(1)}\bigotimes|0\rangle^{(2)} \bigg\rangle\notag\\
		=&
		\langle 0|
		\tilde{\Gamma}_+(x_{N},y_{N})\cdots\tilde{\Gamma}_+(x_{1},y_{1})
		\tilde{\Gamma}_-(z_1,v_1)\cdots\tilde{\Gamma}_-(z_{N},v_{N})
		|0\rangle\notag\\
		=&	\prod_{i,l=1}^{N_1} \frac{1 +  x_{j}z_{l}}{1 -  x_{j}z_{l}} \prod_{j,k=1}^{N_2} \frac{1 + y_{j} v_{k}}{1 - y_{j} v_{k}},
	\end{align}
	which means
	\begin{align}
		\sum_{{\tilde{\pi}^{1} \subseteq [N_{1}, N_{1},\infty]}
			\atop{\tilde{\pi}^{2} \subseteq [N_{2}, N_{2}, \infty]}}
		B_{\tilde{\pi}^{1} }(\{\mathbf{x^1}\}, \{\mathbf{z^1}\}) B_{\tilde{\pi}^{2} }(\{\mathbf{y^2}\}, \{\mathbf{v^2}\})
		=\prod_{i,l=1}^{N_1} \frac{1 +  x_{j}z_{l}}{1 -  x_{j}z_{l}} \prod_{j,k=1}^{N_2} \frac{1 + y_{j} v_{k}}{1 - y_{j} v_{k}}.
	\end{align}
From Eq.$(\ref{scalar-product-S1-S2})$, the Schur $Q$-functions products  of the generating function can be expressed as
	\begin{align}
		\sum_{{ \tilde{\mu}^{1} \subseteq [N_1,\infty]}\atop{\tilde{\mu}^{2} \subseteq [N_2,\infty]}}
		2^{- l(\tilde{\mu}^{1})}	2^{- l(\tilde{\mu}^{2})}	
		Q_{\tilde{\mu}^{1}}\{\mathbf{x^1}\}  Q_{\tilde{\mu}^{1}}\{\mathbf{z^1}\}
		Q_{\tilde{\mu}^{2}}\{\mathbf{y^2}\}  Q_{\tilde{\mu}^{2}}\{\mathbf{v^2}\} 	
		=\prod_{i,l=1}^{N_1} \frac{1 +  x_{j}z_{l}}{1 -  x_{j}z_{l}} \prod_{j,k=1}^{N_2} \frac{1 + y_{j} v_{k}}{1 - y_{j} v_{k}}.
	\end{align}
	
	Consider the double scaling limit assuming that the number of lattice sites $M_{1},M_{2}$ and number of particles $N_{1},N_{2}$ tend to infinity.
	\begin{proposition}
		Setting
		\begin{align}\label{set-z}
			x_{n}=z_{n}=p^{n-\frac{1}{2}}, ~ y_{m}=v_{m}=q^{m-\frac{1}{2}}
			\quad \text{for all } 1 \leq n \leq N_1, ~ 1 \leq m \leq N_2.
		\end{align}
		When $M_{j},N_{j} \to \infty(j=1,2)$, the following equation holds
		
		\begin{align}
			\lim\limits_{{M_{j},N_{j}\to \infty} \atop{j=1,2}	}
			\tilde{S}(\{\mathbf{x^1}\}, \{\mathbf{y^2}\}, \{\mathbf{z^1}\}, \{\mathbf{v^2}\})
			&=
			\sum_{\begin{smallmatrix}\tilde{\pi}^{1}\text{ and }\tilde{\pi}^{2}\text{ are}\\\text{strict plane partitions}\end{smallmatrix}}
			B_{\tilde{\pi}^{1} }(\{\mathbf{x^1}\}, \{\mathbf{z^1}\}) B_{\tilde{\pi}^{2} }(\{\mathbf{y^2}\}, \{\mathbf{v^2}\})\notag\\
			&=	\sum_{\begin{smallmatrix}\tilde{\pi}^{1}\text{ and }\tilde{\pi}^{2}\text{ are}\\\text{strict plane partitions}\end{smallmatrix}}
			2^{p{(\tilde{\pi}^{1})}}2^{p{(\tilde{\pi}^{2})}}
			p^{|\tilde{\pi}^{1}|}q^{|\tilde{\pi}^{2}|}\notag\\
			&=\prod_{n=1}^\infty\left(\frac{1+p^n}{1-p^n}\right)^n
			\prod_{m=1}^\infty\left(\frac{1+q^m}{1-q^m}\right)^m.
		\end{align}
	\end{proposition}
	
	\begin{proof}
		Based on Eqs.$(\ref{Gamma--right})$ and $(\ref{Gamma+-left})$, assume that
		\begin{align}
			&\tilde{\Gamma}_-(p^{i-\frac{1}{2}},q^{i-\frac{1}{2}}) |\tilde{\nu}^1,\tilde{\nu}^2 \rangle\notag\\
			=&\sum_{{\tilde{\nu}^1 \prec \tilde{\pi}^1_{i-1}}	\atop{\tilde{\nu}^2 \prec \tilde{\pi}^2_{i-1}}}
			\prod_{j=1}^{2}
			2^{\#(\tilde{\pi}^j_{i-1}|\tilde{\nu}^j) - l(\tilde{\pi}^j_{i-1})+l(\tilde{\nu}^j)}
			p^{(i-\frac{1}{2})(|\tilde{\pi}^1_{i-1}| - |\tilde{\nu}^1|)}  q^{(i-\frac{1}{2})(|\tilde{\pi}^2_{i-1}| - |\tilde{\nu}^2|)} |\tilde{\pi}^1_{i-1},\tilde{\pi}^2_{i-1} \rangle,  \\
			&\langle \tilde{\nu}^2,\tilde{\nu}^1|\tilde{\Gamma}_+(p^{i-\frac{1}{2}},q^{i-\frac{1}{2}})\notag\\
			=&\sum_{{\tilde{\nu}^1 \prec \tilde{\pi}^1_{-i+1}}	\atop{\tilde{\nu}^2 \prec \tilde{\pi}^2_{-i+1}}}
			\prod_{j=1}^{2}	
			2^{\#(\tilde{\pi}^j_{-i+1}|\tilde{\nu}^{j}) - l(\tilde{\pi}^j_{-i+1})+l(\tilde{\nu}^{j})}
			p^{(i-\frac{1}{2})(|\tilde{\pi}^1_{-i+1}| - |\tilde{\nu}^1|)}    q^{(i-\frac{1}{2})(|\tilde{\pi}^2_{-i+1}| - |\tilde{\nu}^2|)}     \langle \tilde{\pi}^2_{-i+1},\tilde{\pi}^1_{-i+1}|.
		\end{align}
		By using Eqs.$(\ref{Gamma-left})-(\ref{Gamma+left-1})$, we obtain
		\begin{align}
			&\tilde{\Gamma}_-(p^{\frac{1}{2}},q^{\frac{1}{2}})\cdots\tilde{\Gamma}_-(p^{N - \frac{1}{2}},q^{N - \frac{1}{2}})
			|0\rangle \notag\\
			=&	\sum_{{ [N_{1}, \infty]  \ni \tilde{\pi}^{1}_0 \succ \cdots \succ \tilde{\pi}^{1}_{N_{1}} = \emptyset}
				\atop{[N_{2}, \infty]  \ni \tilde{\pi}^{2}_0 \succ \cdots \succ \tilde{\pi}^{2}_{N_{2}} = \emptyset}}
			2^{-l(\tilde{\pi}^1_{0})-l(\tilde{\pi}^2_{0})}
			\prod_{j=1}^{N_1}
			2^{\#(\tilde{\pi}^1_{j-1}|\tilde{\pi}^1_{j})}	p^{(j-\frac{1}{2})(|\tilde{\pi}^1_{j-1}| - |\tilde{\pi}^1_{j}|)}  \prod_{k=1}^{N_2}	
			2^{\#(\tilde{\pi}^2_{k-1}|\tilde{\pi}^2_{k})}
			q^{(k-\frac{1}{2})(|\tilde{\pi}^1_{k-1}| - |\tilde{\pi}^2_{k}|)}  |\tilde{\pi}^{1}_0,\tilde{\pi}^{2}_0\rangle,\label{right-scalar-product} \\
			&\langle 0|
			\tilde{\Gamma}_+(p^{N - \frac{1}{2}},q^{N - \frac{1}{2}})\cdots
			\tilde{\Gamma}_+(p^{\frac{1}{2}},q^{\frac{1}{2}})\notag\\
			=&\sum_{{ \emptyset=\tilde{\pi}^{1}_{-N_{1}} \prec \cdots \prec \tilde{\pi}^{1}_0 \in [N_{1}, \infty]  }
				\atop{\emptyset=\tilde{\pi}^{2}_{-N_{2}} \prec \cdots \prec \tilde{\pi}^{2}_0  \in [N_{2}, \infty]  }}
			2^{-l(\tilde{\pi}^1_{0})-l(\tilde{\pi}^2_{0})}
			\prod_{i=1}^{N_1}
			2^{\#(\tilde{\pi}^1_{-i+1}|\tilde{\pi}^1_{-i})}
			p^{(i-\frac{1}{2})(|\tilde{\pi}^1_{-i+1}| - |\tilde{\pi}^1_{i}|)}  \notag\\
			&\cdot \prod_{l=1}^{N_2}
			2^{\#(\tilde{\pi}^2_{-l+1}|\tilde{\pi}^2_{-l})}	
			q^{(l-\frac{1}{2})(|\tilde{\pi}^2_{-l+1}| - |\tilde{\pi}^2_{l}|)} \langle\tilde{\pi}^{2}_0, \tilde{\pi}^{1}_0 |.\label{left-scalar-product}
		\end{align}
		It follows from the Lemma $\ref{strict-pi}$, Eq.$(\ref{fock-scalar-product})$, Eq.$(\ref{right-scalar-product})$ and Eq.$(\ref{left-scalar-product})$ that
		\begin{align}
			&\langle 0|
			\tilde{\Gamma}_+(p^{N - \frac{1}{2}},q^{N - \frac{1}{2}})\cdots\tilde{\Gamma}_+(p^{\frac{1}{2}},q^{\frac{1}{2}})
			\tilde{\Gamma}_-(p^{\frac{1}{2}},q^{\frac{1}{2}})\cdots\tilde{\Gamma}_-(p^{N - \frac{1}{2}},q^{N - \frac{1}{2}})
			|0\rangle \notag\\
			=&\sum_{{\tilde{\pi}^{1} \subseteq [N_{1}, N_{1},\infty]}	\atop{\tilde{\pi}^{2} \subseteq [N_{2}, N_{2}, \infty]}}		
			2^{p{(\tilde{\pi}^{1})}}2^{p{(\tilde{\pi}^{2})}}
			p^{|\tilde{\pi}^{1}|}q^{|\tilde{\pi}^{2}|}.
		\end{align}
		Meanwhile, Eq.$(\ref{scalar-product-Gamma1-Gamma2-m-wq})$ can be expressed as
		\begin{align}
			&\langle 0|
			\tilde{\Gamma}_+(p^{N- \frac{1}{2}},q^{N - \frac{1}{2}})\cdots\tilde{\Gamma}_+(p^{\frac{1}{2}},q^{\frac{1}{2}})
			\tilde{\Gamma}_-(p^{\frac{1}{2}},q^{\frac{1}{2}})\cdots\tilde{\Gamma}_-(p^{N - \frac{1}{2}},q^{N - \frac{1}{2}})
			|0\rangle \notag\\
			=&\prod_{i,j=1}^{N_1} \frac{1 + p^{j+i-1}}{1 - p^{j+i-1}}
			\prod_{l,k=1}^{N_2} \frac{1 + q^{l+k-1}}{1 - q^{l+k-1}}.
		\end{align}
		Then
		\begin{align}
			\sum_{{\tilde{\pi}^{1} \subseteq [N_{1}, N_{1},\infty]}	\atop{\tilde{\pi}^{2} \subseteq [N_{2}, N_{2}, \infty]}}	
			2^{p{(\tilde{\pi}^{1})}}2^{p{(\tilde{\pi}^{2})}}	p^{|\tilde{\pi}^{1}|}q^{|\tilde{\pi}^{2}|}	
			=\prod_{i,l=1}^{N_1} \frac{1 - p^{l+i-1}}{1 - p^{l+i-1}} \prod_{j,k=1}^{N_2} \frac{1 + q^{j+k-1}}{1 + q^{j+k-1}}.
		\end{align}
		Considering the limit $N_{1},N_{2} \to \infty$, we obtain
		\begin{align}
			&\sum_{\begin{smallmatrix}\tilde{\pi}^{1}\text{ and }\tilde{\pi}^{2}\text{ are}\\\text{strict plane partitions}\end{smallmatrix}}
			2^{p{(\tilde{\pi}^{1})}}2^{p{(\tilde{\pi}^{2})}}
			p^{|\tilde{\pi}^{1}|}q^{|\tilde{\pi}^{2}|}  \notag\\
			=&\lim\limits_{N_{1},N_{2}\to \infty}	
			\left(\prod_{i=1}^{N_1} \frac{1 + p^{i}}{1 - p^{i}}  \cdots  \prod_{i=N_1}^{2N_1} \frac{1 + p^{i}}{1 - p^{i}}\right)
			\left(\prod_{j=1}^{N_2} \frac{1 + q^{i}}{1 - q^{j}}  \cdots  \prod_{j=N_2}^{2N_2} \frac{1 + q^{j}}{1 - q^{j}}\right)\notag\\
			=&\prod_{n=1}^\infty\left(\frac{1+p^n}{1-p^n}\right)^n
			\prod_{m=1}^\infty\left(\frac{1+q^m}{1-q^m}\right)^m.
		\end{align}
	\end{proof}
It is proved that the generating function for BUC plane partitions can be derived  with the double scaling limit.

	\section{Conclusions and discussions}
	This paper is concerned with the  relation between  generalized $i$-boson model and boxed BUC plane partitions.
	With the help of maps $\tilde{\mathcal{M}}$ and $\tilde{\mathcal{M}}^*$,
 the actions of  monodromy matrix operators on the basis vectors generate interlacing  strict boxed   $2$-partition, which  can also be produced by the actions of the neutral fermion vertex operators on the state vectors in $\mathcal{\widetilde{F}}$ and  $\mathcal{\widetilde{F}}^{*}$.
	Meanwhile, the generating function for boxed BUC plane partitions have been  expressed as products of Schur $Q$-functions.
	We also discuss the limiting forms of the scalar product of the   generalized $i$-boson model and obtain the  generating function for  BUC plane partitions with the double scaling limit.
	It is worth studying  the relation	between topological field theories and quantum integrable systems.
	It should be pointed out that the  correspondence
	between the G/G gauged Wess-Zumino-Witten (WZW) model and phase model, as well as that between the G/G gauged WZW-matter model
	and  $q$-boson model, has been established\cite{G/G-phase-model,G/G-q-boson-model}.
	It would be interesting to investigate the correspondence between
	the   generalized $i$-boson model and  the G/G gauged WZW model.

	\section{Acknowledgements}
	This article is dedicated to Professor Ke Wu in Capital Normal University in celebration of his 80th birthday. And this work is supported by the National Natural Science Foundation of China (Grant No.12461048) and Nature Science Foundation of Inner Mongolia Autonomous Region (Grant NO.2023MS1003). The authors gratefully acknowledge the support of Professor Ke Wu and Professor Weizhong Zhao at Capital Normal University, China.


\begin{thebibliography}{99}
		\bibitem{quantum-groups}N. Yu. Reshetikhin, L. A. Takhtadzhyan and L. D. Faddeev, Quantization of Lie groups and Lie algebras, Leningrad Math. J. 1 (1990) 193-225.
		
		\bibitem{conformal-field}V. V. Bazhanov, S. L. Lukyanov and A. B. Zamolodchikov, Integrable structure of conformal field theory, quantum KdV theory and thermodynamic Bethe ansatz, Comm. Math. Phys. 177 (1996) 381-398.
		
		\bibitem{Yang-Baxter-1}A. Yu. Volkov and L. D. Faddeev, Yang-baxterization of the quantum dilogarithm, J. Math. Sci. 88 (1998) 202-207.
		
		\bibitem{Yang-Baxter-2}S. \`{E}. Derkachev and V. P. Spiridonov, Yang-Baxter equation, parameter permutations, and the elliptic beta integral, Russ. Math. Surv. 68 (2013) 1027-1072.
		
		\bibitem{Mac}I. G. Macdonald, Symmetric Functions and Hall Polynomials, Clarendon Press, Oxford 1995.
		
		\bibitem{BKP-1}Y. You, Polynomial solutions of the BKP hierarchy and projective representations of symmetric groups, in Infinite-Dimensional Lie Algebras and Groups, Adv. Ser. Math. Phys. 7 (1989) 449-464.
		
		\bibitem{BKP-2} J. Nimmo, Hall-Littlewood symmetric functions and the BKP equation, J. Phys. A 23 (1990) 751760.
		
		\bibitem{V.G.}V. G. Kac, Infinite Dimensional Lie algebras, Cambridge University Press, Cambridge 1990.
		
		\bibitem{Jimbo1981} E. Date, M. Kashiwara, M. Jimbo and T. Miwa, Transformation groups for soliton equations: Nonlinear integrable systems-classical theory and quantum theory (Kyoto, 1981), World Scientific Publishing, Singapore 1983.
		
		\bibitem{QISM-1}E. K. Sklyanin, L. A. Takhtadzhyan and L. D. Faddeev, Quantum inverse problem method. I, Theor. Math. Phys. 40 (1979) 688-706.
		
		\bibitem{QISM-2}V. Korepin, N. M. Bogoliubov and A. Izergin, Quantum Inverse Scattering Method and Correlation Functions, Cambridge University Press,  Cambridge 1993.
		
		\bibitem{q-boson}N. M. Bogoliubov and R. K. Bullough, A $q$-deformed completely integrable Bose gas model, J. Phys. A 25 (1992) 4057-4071.
		
		\bibitem{boshilunwen}M. Wheeler, Free fermions in classical and quantum integrable models, Ph.D thesis, Univ. Melbourne, arXiv:1110.6703v1.
		
		\bibitem{i-boson}N. H. Jing, Z. J. Li and T. W.  Cai, Correlation functions of charged free boson and fermion systems, J. Stat. Mech. 2020 (2020) 083101.
		
		\bibitem{B97}N. M. Bogoliubov, A. G. Izergin and N. A. Kitanine, Correlators of the phase model, Phys. Lett. A 231 (1997) 347-352.
		
		\bibitem{q-boson-hall-little}N. V. Tsilevich, Quantum inverse scattering method for the $q$-boson model and symmetric functions, Funct. Anal. Appl. 40 (2006) 207-217.
		
		\bibitem{Q-boson}T. Araujo, Q-boson model and relations with integrable hierarchies, Nucl. Phys. B 1006 (2024) 116640.
		
	\bibitem{six} 	R. G. Baxter, Exactly Solved Models in Statistical Mechanics,  Academic press, London 1982.	
		\bibitem{five} N. M. Bogoliubov and C. L. Malyshev, Scalar product of the five-vertex model and complete symmetric polynomials, J. Math. Sci. 284 (2024) 654-664.
		
		\bibitem{MacMahon}P. A. MacMahon, Combinatory Analysis,  Cambridge University Press, Cambridge  1915.
		
		\bibitem{C1} F. Bergeron, G. Labelle and P. Leroux, Combinatorial Species and Tree-Like Structures, Cambridge University Press, Cambridge 1998.
		
		\bibitem{C2} Richard Peter Stanley, Enumerative Combinatorics. Volume 2, Volume 62 of Cambridge Studies in Advanced Mathematics,  Cambridge University Press, Cambridge 1999.	
		
		\bibitem{Statistical-1}A. M. Vershik, Statistical mechanics of combinatorial partitions, and their limit shapes, Funct. Anal. Appl. 30 (1996) 90-105.
		
		\bibitem{Statistical-2} R. J. Baxter, Exactly Solved Models in Statistical Mechanics, Academic Press, New York 1982.
		
		\bibitem{ferm}A. Okounkov, Random matrices and random permutations, Int. Math. Res. Not. 2000 (2000) 1043-1095.
		
		\bibitem{bosonic}N. M. Bogolyubov, Enumeration of plane partitions and the algebraic Bethe anzatz, Theor. Math. Phys. 150 (2007) 165-174.
		
		\bibitem{Schur-process} A. Okounkov and N. Reshetikhin, Correlation function of Schur process with application to local geometry of a random 3-dimensional Young diagram, J. Am. Math. Soc. 16 (2003) 581-603.
		
		\bibitem{shifted-Schur-process}  M. Vuleti\'{c}, Shifted Schur process and asymptotics of large random strict plane partitions, Int. Math. Res. Not. 14 (2007) 043.
		
		\bibitem{Foda06} O. Foda and M. Wheeler, BKP plane partitions, J. High Energ. Phys.  01 (2007) 075.
		
		\bibitem{Foda09} O. Foda, M. Wheeler and M. Zuparic, On free fermions and plane partitions, J. Algebra 321 (2009) 3249-3273.
		
		\bibitem{2-parameter} M. Vuleti\'{c}, A generalization of MacMahon's formula, Trans. Am. Math. Soc. 361 (2009) 2789-2804.
		
		\bibitem{Jing-B}N. H. Jing, Vertex operators and Hall-Littlewood symmetric functions, Adv. Math. 87 (1991) 226-248.
		
		\bibitem{Jing-F}N. H. Jing, Boson-fermion correspondence for Hall-Littlewood polynomials, J. Math. Phys. 36 (1995) 7073-7080.
		
		\bibitem{1-parameter} O. Foda and M. Wheeler, Hall-Littlewood plane partitions and KP, Int. Math. Res. Not. 14  (2009) 2597-2619.
		
		\bibitem{UC-BUC}S. Y. Zhang and Z. W. Yan, UC and BUC plane partitions, Eur. Phys. J. C 371 (2024) 84.
		
		\bibitem{Boxed-q-boson-model}N. M. Bogoliubov, Boxed plane partitions as an exactly solvable boson model, J. Phys. A 38 (2005) 9415-9430.	
		
		\bibitem{Boxed-skew-phase-model}K. Shigechi and M. Uchiyama, Boxed skew plane partition and integrable phase model, J. Phys. A 38 (2005) 10287-10306.	
		
		\bibitem{Four-vertex-Model} N. M. Bogoliubov, Four-vertex model and random tilings, Theor. Math. Phys. 155 (2008) 523-535.
	\bibitem{Six-vertex-Model}		K. Motegi, Factorization of rational six vertex model partition functions,  Nucl. Phys. B 1009 (2024) 116743.
		\bibitem{t-g-phase-model-boxed-UC}S. Y. Zhang, J. Z. Liu and Z. W. Yan, Boxed UC plane partitions and the two-site generalized phase model, arXiv:2601.03941.
		
		\bibitem{BUC}Y. Ogawa, Generalized Q-functions and UC hierarchy of B-Type, Tokyo J. Math. 32 (2009) 350-380.
		
		\bibitem{G/G-phase-model}S. Okuda and Y. Yoshida, G/G gauged WZW model and Bethe Ansatz for the phase model, J. High Energ. Phys. 2012 (2012) 146.
		
		\bibitem{G/G-q-boson-model}S. Okuda and Y. Yoshida, G/G gauged WZW-matter model, Bethe Ansatz for $q$-boson model and Commutative Frobenius algebra, J. High Energ. Phys. 2014  (2014) 3.
		
		
	\end{thebibliography}
\end{document}